\documentclass[12pt]{article}
\usepackage{epsfig,amsfonts,amsthm}
\newcommand{\be}{\begin{equation}}
\newcommand{\ee}{\end{equation}}
\newcommand{\bea}{\begin{eqnarray}}
\newcommand{\eea}{\end{eqnarray}}
\newcommand{\dst}{\displaystyle}
\newcommand{\fr}[2]{\frac{{\dst #1}}{{\dst #2}}}

\renewcommand{\Re}{\mbox{Re }}
\renewcommand{\Im}{\mbox{Im }}
\newcommand{\stolb}[3]{ \left( \begin{array}{c}#1 \\ #2 \\ #3\end{array}\right) }
\newcommand{\stolbik}[2]{ \left( \begin{array}{c}#1 \\ #2 \end{array}\right) }

\newcommand{\lr}[1]{ \langle #1 \rangle}

\newtheorem{theorem}{Theorem}
\newtheorem{proposition}[theorem]{Proposition}

\def\lsim{\mathrel{\rlap{\lower4pt\hbox{\hskip1pt$\sim$}}
    \raise1pt\hbox{$<$}}}         
\def\gsim{\mathrel{\rlap{\lower4pt\hbox{\hskip1pt$\sim$}}
    \raise1pt\hbox{$>$}}}         

\title{General two-order-parameter Ginzburg-Landau model with quadratic and quartic interactions}

\author{I.P. Ivanov\thanks{E-mail: Igor.Ivanov@ulg.ac.be}\\
  {\small Interactions Fondamentales en Physique et en Astrophysique, Universit\'{e} de Li\`{e}ge,} \\
  {\small All\'{e}e du 6 Ao\^{u}t 17, b\^{a}timent B5a, B-4000 Li\`{e}ge, Belgium}\\
  {\small and}\\
  {\small Sobolev Institute of Mathematics, Koptyug avenue 4, 630090, Novosibirsk, Russia}}

\begin{document}
\maketitle

\begin{abstract}
Ginzburg-Landau model with two order parameters appears in many condensed-matter problems.
However, even for scalar order parameters, the most general $U(1)$-symmetric Landau potential
with all quadratic and quartic terms contains 13 independent coefficients
and cannot be minimized with straightforward algebra.
Here, we develop a geometric approach that circumvents this computational difficulty
and allows one to study properties of the model without knowing the exact position
of the minimum.
In particular, we find the number of minima of the potential, classify explicit symmetries
possible in this model, establish conditions when and how these symmetries are spontaneously
broken, and explicitly describe the phase diagram.
\end{abstract}

\section{Introduction}

The Landau theory \cite{Landau,LL5} offers a remarkably economic description
of phase transitions associated with symmetry breaking.
This breaking is described
by an order parameter $\psi$: the high symmetry phase corresponds to
$\psi=0$, while the low symmetry phase is described by
$\psi \not =0$.
Very often, the order parameter can be directly related to
physically observable quantities, such as, for example,
distortion of the crystal lattice or spontaneous magnetization,
The local version of the Landau theory with a coordinate-dependent
order parameter, known as Ginzburg-Landau (GL) theory,
is the basis of the phenomenological theory of superconductivity, \cite{GLoriginalSC}.
For a variety of applications of the Landau theory to various condensed-matter problems,
see e.g. \cite{GL} and references therein.

In order to find if a given system is in its high or low symmetry
phase, one constructs a Landau potential
that depends on the order parameter, and then finds its minimum.
For a complex order parameter, its classic form \cite{Landau,LL5} is
\be
V(\psi) = - a |\psi|^2 + {b \over 2} |\psi|^4 + o(|\psi|^4)\,.\label{landau}
\ee
Near the phase transition the higher order terms $o(|\psi|^4)$ are usually assumed to be negligible.
The values of the coefficients $a$ and $b$ and their dependence on temperature, pressure, etc.
can be either calculated from a microscopic theory, if it is available,
or considered as free parameters in a phenomenological approach.
The phase transition associated with the symmetry breaking takes place
when an initially negative $a$ becomes positive, and the minimum of the potential (\ref{landau})
shifts from zero to
\be
\lr{\psi} = \sqrt{{a \over b}} \, e^{i\alpha}\,,
\ee
with an arbitrary phase $\alpha$.

Many systems are known in which two competing order parameters (OP) coexist.
Among them are the general $O(m)\oplus O(n)$-symmetric models, \cite{EXOmOn},
the models with two interacting $N$-vector OPs with $O(N)$ symmetry, \cite{EXtwoOn};
spin-density-waves in cuprates, \cite{EXspindensity};
competition between antiferromagnetism and superconductivity, \cite{EXantiferrosuper};
${}^4$He with its interplay of crystalline and superfluid ordering, \cite{EXcrystallinesuper};
multicomponent, \cite{EXmulticomponentsuper}, non-conventional two-dimensional, \cite{EXtwodimensional},
spin-triplet $p$-wave, \cite{EXtripletsuper},
and two-gap, \cite{EXtwoband,EXtwogapGurevich}, superconductivity,
with its application to magnetism in neutron stars, \cite{EXneutronstars};
two-band superfluidity, \cite{EXtwobandfluid};
and even mechanisms of electroweak symmetry breaking beyond the Standard Model
such as the two-Higgs-doublet model (2HDM), \cite{EX2HDM}.

To describe such a situation within GL theory,
one constructs a Landau potential similar to (\ref{landau}),
which depends on two order parameters, $\psi_1$ and $\psi_2$.
With scalar order parameters, it can be written generically as
\be
V(\psi_1,\psi_2) = - a_{ij} (\psi^*_i \psi_j) + {1 \over 2} b_{ijkl} (\psi^*_i \psi_j)(\psi^*_k \psi_l)\,,
\quad i,j,k,l=1,2\,.\label{Vijkl}
\ee
The coefficients of this potential can be
considered independent although in each particular
application they might obey specific relations.
One thus arrives at the general two-order-parameter (2OP) GL model
with quadratic and quartic terms.

Once Landau potential (\ref{Vijkl}) is written, the next step is to find
its minimum, i.e. to solve the static homogeneous Ginzburg-Landau equations.
A rather surprising fact is that these equations cannot be solved
with straightforward algebra. Differentiating the Landau potential
with respect to $\psi_i$ leads to a system of coupled algebraic
equations of third order, whose total
degree of algebraic complexity is six,
which makes it impossible to solve in the general case.

In this paper we argue that despite this computational problem, there remains something
that one can learn about the most general two-order-parameter model
in the mean-field approximation: its {\em phase diagram}.
As we will show below, it is possible to classify all the phases according
to the symmetries of the model and properties of the ground state.

This idea is not new. In fact, there exists an extensive
literature dating back to 1980's on minimization of
$G$-invariant potentials with several OPs ($G$ being a group of transformations of OPs), see e.g.
\cite{GL,sartori,Kim} and references therein. These works exploit
the fact that the problem becomes simpler when reformulated in the
orbit space instead of the space of order parameters themselves,
\cite{michel}. This orbit space is naturally sliced into several strata,
which are linked to the allowed phases of the model.
To describe them, one constructs the ring of $G$-invariant polynomials
of the order parameters and finds the Minimal Integrity Basis (MIB)
of this ring. Different strata (i.e. different phases) correspond to
some particular relations among the MIB polynomials.
This general method has been applied to classification of the phases of several
relevant physical systems, for example, to $p$-wave superfluidity
of ${}^3$He, \cite{Bruder}, $D$-wave condensates, \cite{Dwave},
2HDM, \cite{sartori2HDM,gufan2HDM}, and to general Landau models with
multicomponent and even multidimensional order parameters,
\cite{Kimmultidimensional,gufanMulti}.

In the light of this activity, it is somewhat surprising that the
most general GL model with two complex order parameters and
with the most general potential in the form of (\ref{Vijkl}) has never been studied in complete detail.
Here we fill this gap showing that in this case
the analysis can be pushed much farther than in the general situation,
with important physical consequences.

The approach presented here is based on the reparametrization symmetry of the model,
which allows one to establish the Minkowski-space structure
of the orbit space. The minimization problem admits a transparent geometric
interpretation, which leads to several theorems concerning the properties
of the global minimum.
Specific application of this approach to the 2HDM was
given in \cite{mink,minknew}.
Here we analyze the case of two local order parameters
$\psi_1(\vec r)$ and $\psi_2(\vec r)$ in the general context,
which can be relevant also for many condensed-matter problems.

\subsection{Geometric approach vs. Minimal Integrity Basis method}

Let us stress from the very start the essential differences between the geometric
analysis of the present paper and the standard approach
based on the Minimal Integrity Basis (MIB) technique.

The first difference lies in the scope of these two approaches.
The MIB leads to interesting results in the cases when the
potential is invariant under a non-trivial group $G$ of transformation of the order parameters.
The larger $G$, the richer is the spectrum of possible patterns
of its spontaneous violation. In particular, MIB methods has nothing to say
if $G$ is the trivial group.

In a typical situation one takes a highly symmetric $G$-invariant potential
constructed from powers of several multidimensional order parameters up to a certain degree,
builds various invariants, finds the ones that form MIB, and classifies the possible phases
according to relations among these invariants.
This approach is rather general in a sense that it can be applied, in principle,
to any number and any dimensions of order parameters. However, because of $G$-invariance,
the potentials usually contain very few terms.

The geometric approach presented here is limited to the particular case
of two complex scalar or vector order parameters, and to the
fourth-degree potentials. However, within these restrictions,
we manage to work out the {\em most general} model with
all possible types of the OP interactions.
The only symmetry that we impose is the $U(1)$-symmetry of the free energy density,
which is a reasonable choice from the physical point of view.
In this aspect, our analysis is more general than the MIB approach:
we just take two OPs, construct the free energy density in its full complexity,
and study {\em everything} that can ever happen in this model.
That is, we analyze all possible symmetry groups $G$
and all possible patterns of symmetry breaking.

The second difference concerns the procedures and the results of these two approaches.
In a situation when several phases are possible, one wants to know which
phase corresponds to the ground state of the model ({\em i.e.} which phase is stable).
In the usual MIB method one can do nothing but explicitly solve
the algebraic equations and check the minimum conditions.
This can be done only if the equations are simple enough,
which in turns happen when the free energy density is simple.
Thus, only sufficiently symmetric potentials are fully tractable with the MIB method.
Examples cited in \cite{gufan2HDM,gufanMulti} are precisely of this type.

In terminology suggested by \cite{toledanodmitriev}, one should distinguish between the
{\em angular problem} (classifying all the allowed phases) and the {\em radial problem}
(actually finding the position of the absolute minimum of a given potential).
MIB methods allow one to solve the angular, but not the radial problem.

In the case of the most general 2OP GL model, with its large number of free parameters,
this algebra cannot be worked out explicitly.
One ends up with a general algebraic equation of sixth order, which one cannot solve analytically.
Thus, one is unable to solve the problem of minimization of the potential using
only the MIB formalism.

The strongest point of the present geometric approach is that we {\em avoid}
solving these equations and nevertheless we rigorously prove several statements
about the ground state of the model. In other words, we study the properties
of the absolute minimum {\em without} solving the radial problem.
This is especially useful for the case of the {\em smallest} possible group $G$,
for which the MIB technique becomes redundant.
Thus, the geometric approach presented here is neither a particular case
nor an improvement of the MIB method, but is complementary to it.\\

The structure of the paper is the following.
In Section~\ref{section-formalism} we introduce the formalism and derive a very compact
expression for the free energy functional.
The extrema of the Landau potential cannot be found with straightforward algebra,
so in Section~\ref{section-minima} we develop geometric tools that allow us to find
the number of extrema and minima of the potential.
Section~\ref{section-weak} is devoted to the special case of a potential stable in a weak sense.
Then, in Section~\ref{section-symmetries} we give full classification of explicit symmetries
of the model and derive conditions when and how these symmetries are spontaneously broken.
All this allows us to describe in Section~\ref{section-phase-diagram} the phase diagrams
of the model, listing the phases according to the number of minima and symmetries.
Here, we also discuss phase transitions and argue that critical properties, too, can
be calculated in geometric terms.
Section~\ref{section-examples} contains analysis of several simple cases, which provide illustration
of the general approach.
In the short Section~\ref{section-solitons} we outline conditions when solitons appear in this model.
In Section~\ref{section-multi} we outline characteristic features
of the general GL model with two complex $N$-vector order parameters, and finally
in Section~\ref{section-conclusions} we draw our conclusions.
Appendices provide some mathematical details and derivations.

\section{Formalism}\label{section-formalism}

In the main part of the paper we will assume that $\psi_i(\vec r)$
are just complex numbers;
modifications in the case of more complicated OPs will be discussed
in Section~\ref{section-multi}.
Throughout the paper we also assume that the absolute values of $|\psi_i|$
are not bounded from above.

Let us consider the free-energy density
in the most general globally $U(1)$-invariant 2OP GL model
containing all possible quadratic and quartic terms in the potential:
\be
F = K + V_2 + V_4\,.\label{freeenergy}
\ee
It is a sum of the gradient term $K$,
\be
K = \kappa_1 |\vec{D} \psi_1|^2 + \kappa_2 |\vec{D} \psi_2|^2
+ \kappa_3 (\vec{D}\psi_1)^*(\vec{D} \psi_2) + \kappa_3^* (\vec{D}\psi_2)^*(\vec{D} \psi_1)
\,,\label{gradient}
\ee
where $\vec{D}$ is either $\vec{\nabla}$ or the covariant derivative,
and the Landau potential
\bea
V_2&=&- a_1|\psi_1|^2 - a_2|\psi_2|^2
- a_3 (\psi_1^*\psi_2) - a_3^{*} (\psi_2^*\psi_1)\,;\label{potential}\\[2mm]
V_4&=&
\fr{b_1}{2}|\psi_1|^4 + \fr{b_2}{2}|\psi_2|^4
+ b_3 |\psi_1|^2|\psi_2|^2 + \left[\fr{b_4}{2}(\psi_1^*\psi_2)+
b_5 |\psi_1|^2 + b_6 |\psi_2|^2\right](\psi_1^*\psi_2) +{\rm
c.c.}\nonumber
\eea
Free energy density (\ref{freeenergy}) contains $4+4+9=17$ free parameters: real
$\kappa_1,\, \kappa_2,\, a_1,\, a_2,\, b_1,\, b_2,\, b_3$ and
complex $\kappa_3,\, a_3,\, b_4,\, b_5,\, b_6$.

By construction, the free energy
remains invariant under the $U(1)$ group of simultaneous multiplication
of $\psi_1$ and $\psi_2$ by the same global phase factor.
We do not consider terms that violate this symmetry,
such as $\psi_1^2 + (\psi_1^2)^*$.

Note that potential (\ref{potential}) contains quartic terms
such as $|\psi_1|^2(\psi_1^*\psi_2)$ that mix
$\psi_1$ and $\psi_2$, which are usually absent
in many particular applications of the 2OP GL model.
However, in certain cases such terms appear,
as it happens in the dirty limit of a two-gap superconductor,
see e.g. \cite{EXtwogapGurevich}.

We stress that in our approach it is essential that
we include all possible terms from the very beginning.

\subsection{Reparametrization symmetry}\label{section-reparametrization}

From the physical point of view, the order parameters $\psi_1$ and $\psi_2$
can be of the same (as in two-gap superconductors) or of different nature
(as in the case of superfluid/crystalline ordering interplay).
However, one can always make OPs dimensionless, and
once the free energy density (\ref{freeenergy}) is constructed and the problem
of its minimization is posed, the physical nature of the OPs becomes irrelevant.

One can then view OPs $\psi_1$ and $\psi_2$ as components
of a single complex 2-vector $\Phi$:
$$
\Phi = \stolbik{\psi_1}{\psi_2}\,,
$$
and consider transformations that mix $\psi_1$ and $\psi_2$. These are assumed to be local
transformations, i.e. they mix $\psi_i(\vec r)$ taken at the same point $\vec r$.

We start with the observation that the most general free energy density (\ref{freeenergy})
retains its generic form under any regular linear transformation between $\psi_1$ and $\psi_2$.
In other words, transformation
\be
\Phi \to \Phi' = T\cdot \Phi\,,\quad \mbox{with any }T \in GL(2,C)\,,\label{transform}
\ee
again leads to (\ref{freeenergy}) but with reparametrized coefficients:
\be
\{\kappa_i,\, a_i,\, b_i\} \to \{\kappa_i',\, a_i',\, b_i'\} = \tau(\kappa_i,\, a_i,\, b_i)\,.
\label{transformtau}
\ee
The explicit link between $T$ and $\tau$ will be given below.

Since any $T\in GL(2,C)$ is invertible, so is $\tau$.
Therefore, if (\ref{transform})
is accompanied by the transform $\tau^{-1}$ of the coefficients, then one arrives at
{\em exactly the same} expression for the free energy as before.

If one considers the free energy only, then the physical observables, such as the depth
of the Landau potential at the minimum
and the eigenvalues of the second derivative matrix of the potential (the hessian)
can be expressed in terms of the coefficients $\{\kappa_i,\, a_i,\, b_i\}$ only.
Therefore, the models $(\Phi,\{\kappa_i,\, a_i,\, b_i\})$
and $(\Phi',\{\kappa_i^\prime,\, a_i^\prime,\, b_i^\prime\})$ have the same sets of observables.
In other words, reparametrization transformations do not change the physical content
of a given model; they only affect the way we look at it.
Thus, we have a {\em reparametrization freedom}
in this problem, with the reparametrization group $GL(2,C)$.

The general linear group $GL(2,C)$ is an 8-dimensional Lie group. It can be written as
\be
GL(2,C) = \mathbb{C}^*\times SL(2,C)\,,
\ee
where $\mathbb{C}^*$ is the group of all multiplications of $\Phi$ by a non-zero complex number.
Due to the $U(1)$-invariance of the free energy,
multiplication of $\Phi$ by an overall phase factor
induces the identity transformation of the coefficients,
while the 7-dimensional factorgroup $GL(2,C)/U(1)$ induces non-trivial transformations $\tau$.
Thus, the 17-dimensional space of coefficients (i.e. the space of
all possible 2OP GL models) becomes sliced into 7-dimensional regions of essentially identical
models linked by all possible $\tau$.
The space of distinct physical situations is described by the corresponding 10-dimensional factorspace.

\subsection{Orbit space}

Let us now introduce the four-vector $r^\mu = (r_0,\,r_i) = (\Phi^\dagger \sigma^\mu \Phi)$ with components
\be
r_0 = (\Phi^\dagger \Phi) = |\psi_1|^2 + |\psi_2|^2\,,\quad
r_i = (\Phi^\dagger \sigma_i \Phi) =
\stolb{2\Re (\psi^*_1 \psi_2) }{2\Im (\psi^*_1 \psi_2) }{|\psi_1|^2 - |\psi_2|^2}\,.
\label{ri}
\ee
Here, index $\mu = 0,1,2,3$ refers to the components in the internal space
and has no relation to the space-time.
Multiplying $\psi_i$ by a common phase factor does not change $r^\mu$,
so each $r^\mu$ uniquely parametrizes a single $U(1)$-orbit in the $\psi_i$-space.
The $U(1)$-invariant free energy (\ref{freeenergy})
can be also defined in this $1+3$-dimensional {\em orbit space}.

The $SL(2,C) \subset GL(2,C)$ group of transformations of $\Phi$ induces
the proper Lorentz group $SO(1,3)$ of transformations of $r^\mu$.
This group includes 3D rotations of the vector $r_i$ as well as ``boosts''
that mix $r_0$ and $r_i$, so
the orbit space gets naturally equipped with the {\em Minkowski space structure}
with metric diag$(1,\,-1,\,-1,\,-1)$.
We stress again that the words ``Minkowski space'' and ``Lorentz group'' always
refer to the internal space, not to the usual space-time.

Since the order parameters $\psi_1$ and $\psi_2$ are just complex numbers,
direct calculation shows that
\be
r^\mu r_\mu \equiv r_0^2 - r_i^2 = 0\,. \label{r2}
\ee
Then, since $r_0>0$ and since the values of $r^\mu$ are not restricted from above,
the orbit space of the 2OP GL model is given by the forward lightcone $LC^+$ in the Minkowski space.
As should be expected, the reparametrization group in the orbit space, $SO(1,3)$,
leaves the orbit space invariant.

Analogously to $r^\mu$, one can also introduce
\be
\rho^\mu \equiv (\vec{D} \Phi)^* \sigma^\mu (\vec{D} \Phi)\,.
\ee
Obviously, the reparametrization transformation laws of $\rho^\mu$ are the same as for $r^\mu$.

All this allows us to rewrite the free energy (\ref{freeenergy}) in a very compact form:
\be
F = K_\mu \rho^\mu - A_\mu r^\mu + {1\over 2} B_{\mu\nu} r^\mu r^\nu\,,\label{freeenergy2}
\ee
with
\bea
K^\mu &=& {1\over 2}\left(\kappa_1+\kappa_2,\, -2\Re \kappa_3,\,
2\Im \kappa_3,\, -\kappa_1+\kappa_2\right)\,,\nonumber\\[2mm]
A^\mu &=& {1\over 2}\left(a_1+a_2,\, -2\Re a_3,\,
2\Im a_3,\, -a_1+a_2\right)\,,\label{KAB}\\[2mm]
B^{\mu\nu} &=& {1\over 2}\left(\begin{array}{cccc}
{b_1+b_2 \over 2} + b_3 & -\Re(b_5 + b_6)
    & \Im(b_5 + b_6) & -{b_1-b_2 \over 2} \\[1mm]
-\Re(b_5 + b_6) & \Re b_4 & -\Im b_4 & \Re(b_5 - b_6) \\[1mm]
\Im(b_5 + b_6) & -\Im b_4 & -\Re b_4 & -\Im(b_5 - b_6) \\[1mm]
 -{b_1-b_2 \over 2} & \Re(b_5 - b_6) & -\Im(b_5 - b_6)
 & {b_1+b_2 \over 2} - b_3
\end{array}\right)\,.\nonumber
\eea
Note that due to (\ref{r2}), definition of the matrix $B^{\mu\nu}$ has one degree of freedom,
since $B^{\mu\nu}$ and $\tilde{B}^{\mu\nu} = B^{\mu\nu} - C g^{\mu\nu}$ with any $C$ are equivalent.

The quantities $K^\mu$, $A^\mu$ and $B^{\mu\nu}$ transform as four-vectors and
a four-tensor, respectively. This provides the explicit link between transformations
$T$ and $\tau$ mentioned in Sect.~\ref{section-reparametrization}.
For convenience, we collect in Appendix~\ref{section-app-B} some basic facts
concerning the manipulation of $B^{\mu\nu}$.

\subsection{Properties of the coefficients}\label{section-properties}

General physical requirements place restrictions on possible $K_\mu$ and $B_{\mu\nu}$.

First, the requirement that very large wavevector oscillations must increase
not decrease the free energy implies that $K_\mu$ lies inside the future lightcone:
$K_0>0$, $K_\mu K^\mu >0$. This condition remains true under
an arbitrary $SO(1,3)$ transformation.

Second, we require that the potential is bounded from below in the entire $\psi_i$-space.
Since the potential is build of quadratic and quartic terms, $V=V_2+V_4$,
this can be achieved in two cases (here we coin the terminology of \cite{nachtmann},
where the stability of the Higgs potential in 2HDM was analyzed):
\begin{itemize}
\item
the potential is stable in a {\em strong} sense, if $V_4$ increases in all directions
in the $\psi_i$-space;
\item
the potential is stable in a {\em weak} sense, if $V_4$ is non-decreasing in all directions
in the $\psi_i$-space, and $V_2$ increases along the flat directions of $V_4$.
\end{itemize}
Let us focus on the case of the potentials stable in a strong sense;
the case of the potential stable in a weak sense will be considered in Section~\ref{section-weak}.
The requirement that $V_4$ is positive definite in the entire $\psi_i$-space
means that the quadratic form $B_{\mu\nu} r^\mu r^\nu$ is positive-definite
on the future lightcone $LC^+$. In Appendix~\ref{section-app-positive} we prove
that this is equivalent to the statement that $B_{\mu\nu}$ is diagonalizable
by an $SO(1,3)$ transformation and after diagonalization it takes form
\be
B^{\mu\nu} =
\left(\begin{array}{cccc}
B_0 & 0 & 0 & 0\\
0 & -B_1 & 0 & 0\\
0 & 0 & -B_2 & 0\\
0 & 0 & 0 & -B_3 \end{array}\right)\quad \mbox{with} \quad B_0 > B_1, B_2, B_3\,.\label{Bi}
\ee
We will refer to $B_0$ as the ``timelike'' eigenvalue of $B_{\mu\nu}$
and $B_i$, $i=1,2,3$, as its ``spacelike'' eigenvalues.
The sign minus in front of the spacelike eigenvalues is the result of the Minkowski-space metric,
see Appendix~\ref{section-app-B}.
The degree of freedom in the definition of $B^{\mu\nu}$ amounts to shifting all the eigenvalues by the same
constant and does not affect the inequalities (\ref{Bi}).
However it can be used to manipulate the signs of the eigenvalues.

Finding the eigenvalues of $B^{\mu\nu}$ explicitly in terms of $b_i$
requires solution of a fourth-order characteristic equation,
which constitutes one of the computational difficulties
of straightforward algebra.
We reiterate that in our analysis we never use these explicit expressions.
Our analysis relies only on the fact that the eigenvalues
are real and satisfy (\ref{Bi}).

\section{Minima of the Landau potential}\label{section-minima}

Having introduced the formalism that allows us treat the most general 2OP GL model,
let us proceed to the task of minimization of the free energy functional.
We do not consider here the effects of non-trivial boundary conditions,
so we are looking for homogeneous solutions $\psi_i(\vec r) = \lr{\psi_i}$
that minimize the Landau potential (\ref{potential}).

As mentioned in the Introduction, straightforward algebra is of little help for
the minimization problem, since the resulting
system of coupled equations $\partial V/\partial \psi_i = 0$ cannot be solved
in the general case.
However one can still learn much about the
ground state of the general 2OP GL model without finding its location explicitly.
In this paper we will provide, in particular, answers to the following questions:
\begin{itemize}
\item
How many extrema does the potential with given parameters have? How many of them are minima?
\item
Can the global minimum be degenerate and when does it happen?
\item
When does the global minimum spontaneously break an explicit symmetry of the potential?
\item
What is the phase diagram of the model?
What phase transitions can take place during continuous change of the coefficients of the model?
\end{itemize}

\subsection{Number of extrema}\label{section-extrema}

Let us start with the number of extrema of a generic Landau potential.
In order to find an extremum of $V$ lying on the future lightcone $LC^+$, one can use the standard Lagrange multiplier
method. In this case one needs to introduce only one Lagrange multiplier $\lambda$, which leads to the following
system
\be
\left\{
\begin{array}{l}
B_{\mu\nu} \lr{r^\nu} - \lambda \lr{r_\mu} = A_\mu\,,\\
\lr{r^\mu} \lr{r_\mu} = 0\,.
\end{array}\right.\label{lagrange}
\ee
Here, $\lr{r^\mu}$ labels the position of an extremum.
To avoid cumbersome notation, we omit $\lr{\dots}$ in this subsection.

To establish how many solutions system (\ref{lagrange}) has,
consider the $B_{\mu\nu}$-diagonal frame
(we remind that for a potential stable in a strong sense such a frame always exists).
Then the first line in (\ref{lagrange}) takes form
\be
(B_0 - \lambda) r_0 = A_0\,, \qquad (B_i - \lambda) r_i = A_i\,.
\ee
Rewriting $r_i = r_0 n_i$, where $n_i$ is a unit 3-vector, and eliminating $\lambda$,
one obtains
\be
\left[A_0 - (B_0-B_i)r_0\right]n_i = A_i\,.\label{lagrange2}
\ee
These three equations are coupled via the condition $|\vec n|=1$.
Consider the l.h.s. of (\ref{lagrange2}) at fixed $r_0$ and all the unit vectors $n_i$.
It parametrizes an ellipsoid with semiaxes
\be
A_0 - (B_0-B_1)r_0\,,\quad A_0 - (B_0-B_2)r_0\,,\quad A_0 - (B_0-B_3)r_0\,.\label{semiaxes}
\ee
Now imagine how this ellipsoid changes if $r_0$ increases from zero to infinity.
Let us for simplicity assume that the eigenvalues of $B^{\mu\nu}$ are distinct and $B_1<B_2<B_3$.

Assume first that $A_0>0$.
Then, at $r_0 = 0$, Eq.~(\ref{lagrange2}) parametrizes a sphere with radius $A_0$.
As $r_0$ increases, it turns into a continuously shrinking ellipsoid with semiaxes (\ref{semiaxes}).
At
$$
r_0= r_0^{(1)} \equiv {A_0 \over B_0-B_1}
$$
this ellipsoid collapses to the interior of a planar ellipse with semiaxes
$$
A_0{B_2-B_1 \over B_0-B_1}\,,\quad A_0{B_3-B_1 \over B_0-B_1}\,,
$$
orthogonal to the first axis.
As $r_0$ increases further, this ellipse returns to an ellipsoid with two shrinking and one growing
semiaxes, and at
$$
r_0= r_0^{(2)} \equiv {A_0 \over B_0-B_2}
$$
is collapses again to a flat ellipse with semiaxes
$$
A_0{|B_1-B_2| \over B_0-B_2}\,,\quad A_0{B_3-B_2 \over B_0-B_2}\,,
$$
orthogonal to the second axis. Further on, at $r_0=r_0^{(3)}$ it collapses
to an ellipse orthogonal to the third axis, and for even larger values of $r_0$
this ellipsoid increases infinitely.

\begin{figure}[!htb]
   \centering
\includegraphics[width=6cm]{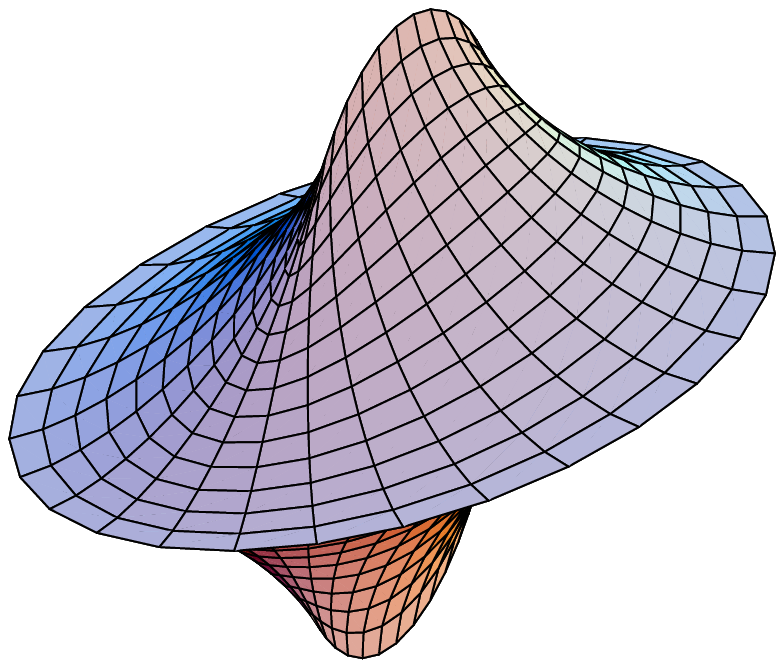}
\includegraphics[width=6cm]{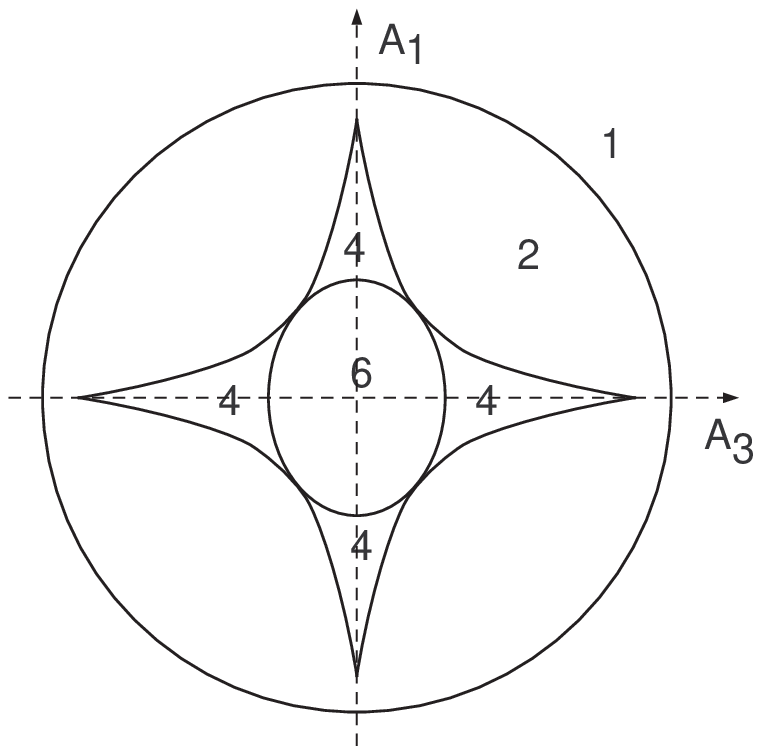}
\caption{(Color online) Left: The envelope of ellipsoids for $r_0^{(1)}<r_0<r_0^{(2)}$.
Right: $(A_1,A_3)$-section of the caustic surfaces in the $A_i$-space
and of the sphere with radius $A_0$. The number of solutions of Eq.~(\ref{lagrange2})
is indicated for each region.}
   \label{fig-caustic-region}
\end{figure}

For each $r_0$ interval, the ellipsoid sweeps a certain region in the three-dimensional space.
\begin{itemize}
\item
During the first stage, $0<r_0<r_0^{(1)}$, it sweeps the interior of the sphere of radius $A_0$,
passing through each point exactly once.
\item
during the second stage, $r_0^{(1)}<r_0<r_0^{(2)}$, it sweeps a certain region, bounded
by the caustic surface shown in Fig.~\ref{fig-caustic-region}, left.
It can be shown that each point inside this region is swept exactly {\em twice}.
\item
during the third stage, $r_0^{(2)}<r_0<r_0^{(3)}$, it sweeps twice
a similar caustic region, but oriented differently;
\item
Finally, during the fourth stage, $r_0 >r_0^{(3)}$, it sweeps once the entire 3D-space.
\end{itemize}
Note that the appearance of caustic regions in the potential extremization problem is natural,
since this problem is known to exhibit some catastrophe theory phenomena, see e.g. \cite{GL}.

Returning to the system (\ref{lagrange2}), which equates the l.h.s. to the 3-vector $A_i$, one sees
that in order get the number of solutions of (\ref{lagrange2}) without finding them explicitly,
one simply has to check whether $A_i$ falls inside these regions.
Fig.~\ref{fig-caustic-region}, right, illustrates this statement. It shows regions with
different numbers of extrema on the $(A_1,A_3)$ plane for $A_2=0$ and some $A_0>0$.

For the second possibility, $A_0<0$, the situation is much simpler.
At $r_0 = 0$ we again start with the sphere of radius $|A_0|$.
As $r_0$ grows, it turns into an ellipsoid with growing semiaxes
(which is due to $B_0-B_i>0$), and it sweeps once the entire space outside the sphere.

The size of the above 3D regions is proportional to $A_0$.
Therefore, in the $1+3$-dimensional space of four-vectors $A^\mu$, they define
the corresponding conical regions starting from the origin.
Therefore, the number of extrema of the potential depends on where the four-vector $A^\mu$ lies:
\begin{itemize}
\item
If $A^\mu$ lies inside the {\em past lightcone} $LC^-$ (i.e. $A_0<0$ and $|\vec{A}|<|A_0|$),
then system (\ref{lagrange}) has {\em no solution}.
In this case the quadratic term of the potential, $-A_\mu r^\mu$,
increases in all directions in the $\psi_i$-space.
The only extremum of the potential is the global minimum at the origin,
which corresponds to the high-symmetry ground state of the model.
\item
If $A^\mu$ lies outside $LC^-$, then {\em at least one} non-trivial solution exists.
If $A_0<0$ in the $B^{\mu\nu}$-diagonal basis,
i.e. $A^\mu$ still lies in the lower hemispace,
then this solution is unique and is the global minimum of the potential.
\item
If $A^\mu$ lies inside $LC^+$, then {\em at least two} non-trivial extrema exist.
\item
If $A^\mu$ lies inside one or both caustic cones defined above,
then {\em two additional extrema per cone} appear.
\end{itemize}
In total, there can be up to six non-trivial extrema of the potential in the orbit space.
This result was also found independently in \cite{nachtmann} with a more traditional analysis
of the Higgs potential of 2HDM.
The largest number of extrema is realized in situations when $A_0>0$ and $A_i$ is sufficiently small,
so that $A^\mu$ lies inside both caustic cones.

Special care must be taken when $r_0$ of an extremum is exactly equal to one of the values $r_0^{(i)}$.
Then the sequence of intersections of the ellipsoid with a given point $A_i$ changes, but
the overall counting rules given above remain the same. As we will see later, this situation
corresponds to spontaneous violation of a discrete symmetry.

\subsection{Number of local minima}\label{section-local-minima}

In general, the above construction cannot distinguish a local minimum from a saddle point
or a maximum, so other methods must be used to establish the number of local minima.

First of all, let us note that potential (\ref{potential}) with restrictions (\ref{Bi}) cannot have non-trivial
maxima \cite{michel,mink}.
This can be easily seen by drawing any ray in the $\psi_i$-space from the origin
and observing that the potential along this ray can be written as
$\alpha |\psi|^2 +\beta |\psi|^4$ with $\beta>0$.
This function can never have a non-trivial maximum. Thus, the problem reduces to distinguishing minima
from saddle points (in the orbit space).

Take a generic extremum of the potential in the $\psi_1$, $\psi_2$
space, and calculate the second derivative matrix of the potential
(the hessian) at this point
\be
(\Omega^2)_{\alpha\beta} =
{\partial^2 V \over \partial \phi_\alpha \partial
\phi_\beta}\,.\label{omega}
\ee
Here, $\phi_\alpha$ are the four
real degrees of freedom, real and imaginary parts of $\psi_1$ and
$\psi_2$:
\be
\psi_1 = \phi_1 + i \phi_2\,; \quad \psi_2 = \phi_3 + i \phi_4\,.\label{phialpha}
\ee
 We will refer to the eigenvalues of this matrix as
``eigenfrequencies'', $\omega^2_a$. Due to the $U(1)$-invariance
of the potential, it always has one flat direction with zero
eigenfrequency (one Goldstone mode), while among the other three,
there is at least one positive eigenfrequency. Let us call the
signs of these three eigenfrequencies (i.e. $+++$, $++-$, or
$+--$) the {\em signature} of the hessian.

Among the four degrees of freedom in the $\psi_i$-space, three correspond to variations in the orbit space,
i.e. to shifts of the point $r^\mu$ on $LC^+$ away from the extremum.
If the extremum is not at the origin, then these shifts are linear functions of the
shifts in the $\psi_i$-space, and the Jacobian corresponding to this transformation is regular.
Indeed, with the notation (\ref{phialpha}), one gets:
\be
{1 \over 2} {\partial r^\mu \over \partial \phi_\alpha} =
\left(\begin{array}{cccc}
\phi_1 & \phi_2 & \phi_3 & \phi_4 \\
\phi_3 & \phi_4 & \phi_1 & \phi_2 \\
\phi_4 & -\phi_3 & -\phi_2 & \phi_1 \\
\phi_1 & \phi_2 & - \phi_3 & - \phi_4
\end{array}
\right)\,.\label{jacobian}
\ee
If $\Phi_1$ is non-zero, then this matrix has one and only one zero eigenvalue,
with the corresponding right eigenvector $(-\phi_2,\, \phi_1,\, -\phi_4,\, \phi_3)$ being the Goldstone
mode. This can be seen most easily in the frame where $r^\mu \propto (1,0,0,1)$,
implying $\phi_3=\phi_4=0$ (obviously, such a frame always exists for any $r^\mu$).
The signature of the hessian, therefore, is the same in the $\psi_i$-space and
in the orbit space.

For an extremum to be minimum, its signature must be $+++$.
However, since the explicit expressions for $\lr{r^\mu}$ cannot be given
in the general case, checking this explicitly at each extremum is also problematic.
One can circumvent this computational difficulty using the following Proposition:
\begin{proposition}\label{prop-zones}
For each extremum, the hessian remains signature-definite
within each conical region described in the previous subsection.
\end{proposition}
\begin{proof}
Let us fix $B^{\mu\nu}$ and move $A^\mu$ continuously in the parameter space,
keeping it strictly inside one of the regions described in the previous subsection.
Let us pick up an extremum and follow how it changes when $A^\mu$ moves.
Its position, its depth as well as the eigenfrequencies are algebraic functions of
the components of $A^\mu$ and therefore they also change continuously.
So, if the hessian changes signature along at the endpoints of some $A^\mu$ path,
then there exists a point, at which one of the eigenfrequencies is zero.
Thus, the expansion of the potential at this point starts from the third or fourth
order term, and this points corresponds to merging of two or three simple extrema,
respectively.

But such a merging cannot happen for $A_\mu$ lying strictly inside the caustic regions.
Indeed, if two sufficiently close points $r_a^\mu$ and $r_b^\mu$ are both extrema of the potential,
then their respective zeroth components $r_{a0}$ and $r_{b0}$ are also close,
so, the intersection points of the corresponding ellipsoids lie close to the
boundary of a caustic region.
In the limit $r_a^\mu \to r_b^\mu$, the intersection points, $A_\mu$ being among them,
approach the envelope (loosely speaking, the envelope can be viewed
as the locus of intersections of the ``successive'' ellipsoids).
\end{proof}
So, since the signature of the hessian remains the same for all $A_\mu$
inside some region, one can select some representative $A_\mu$, calculate
the signature of the hessian for it, and then extrapolate the results for all the
points inside this region.

Let us now calculate the number of minima inside the innermost region of the $A^\mu$ space,
see Fig.~\ref{fig-caustic-region}, right.
For the representative point in this region, $A^\mu = (A_0,\,0,\,0,\,0)$,
calculations can be easily done explicitly.
Indeed, it follows from (\ref{lagrange2}) that there are three pairs of extrema
at $r_0 = r_0^{(1)}$, $r_0^{(2)}$, and $r_0^{(3)}$:
\be
r_0^{(1)}(1,\,\pm 1,\,0,\,0)\,,\quad
r_0^{(2)}(1,\,0,\,\pm 1,\,0)\,,\quad
r_0^{(3)}(1,\,0,\,0,\,\pm 1)\,.\label{threepairs}
\ee
Again, let us order the eigenvalues $B_i$, $B_1<B_2<B_3$ and
expand the potential near the point $\lr{r^\mu}_+ = r_0^{(3)}(1,\,0,\,0,\,1)$.
If $r^\mu = r_0 (1,\,\sin\theta\cos\phi,\,\sin\theta\sin\phi,\,\cos\theta) = \lr{r^\mu}_+ + \delta r^\mu$,
then
\bea
V &=& - A_0 r_0 + {1 \over 2} r_0^2\left(B_0 - B_1 \sin^2\theta\cos^2\phi - B_2 \sin^2\theta\sin^2\phi
- B_3 \cos^2\theta\right) \nonumber\\
&\approx & - {A_0^2 \over 2(B_0 - B_3)} + {B_0-B_3 \over 2}(\delta r_0)^2 +
{A_0^2 \theta^2 \over 2}{(B_3-B_1)\cos^2\phi + (B_3-B_2)\sin^2\phi\over (B_0 - B_3)^2}\,.
\label{expansion}
\eea
Here, $\delta r_0$ and $\theta$ are small, while $\phi$ can be arbitrary.
Since $B_3$ is the largest spacelike eigenvalue, this point is a minimum,
and so is the other extremum of this pair, $\lr{r^\mu}_- = r_0^{(3)}(1,\,0,\,0,\,-1)$.
The same calculation for the extrema at $r_0 = r_0^{(1)}$ and $r_0^{(2)}$,
shows that they are saddle points.
Thus, we find that for $A_\mu$ lying in the innermost region, the potential
has two separate minima and four separate saddle points in the orbit space.

As $A_\mu$ moves out of this region, the number of minima does not increase.
Indeed, one can show that crossing the caustic surface at a generic point
leads to disappearance of two saddle points or of one saddle point and one minimum,
but it cannot, for example, lead to disappearance of three saddle points and appearance
of a new minimum. This can be also verified with the straightforward calculations
similar to (\ref{expansion}) by selecting points $A_\mu$ lying on the axes
(for this choice, all the extrema can be also studied with explicit algebra).
Therefore, we arrive at the following Proposition:
\begin{proposition}\label{prop-two-minima}
The most general quadratic plus quartic potential with two order parameters
can have at most two distinct local minima in the orbit space.
\end{proposition}

\subsection{The principal caustic cone}\label{section-principal}

\begin{figure}[!htb]
   \centering
\includegraphics[width=6cm]{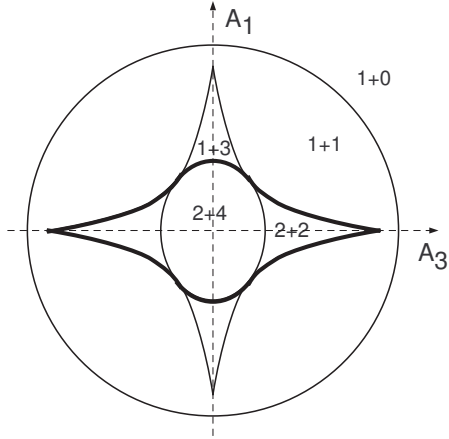}
\caption{The same as in Fig.~\ref{fig-caustic-region}, right,
but with number of minima and saddle points shown separately
as $N_{minima}+N_{saddle}$. Thick lines show the section of the principal caustic
cone.}
   \label{fig-zones2}
\end{figure}

In Section~\ref{section-extrema} we showed that for a generic potential
there exist two caustic cones in the $A^\mu$ space. If $B_1<B_2<B_3$, they corresponds to
$r_0^{(1)} \le r_0 \le r_0^{(2)}$ and $r_0^{(2)} \le r_0 \le r_0^{(3)}$, respectively.
The analysis of Section~\ref{section-local-minima}
shows that these two caustic cones play different role.
It is the second cone, with $r_0^{(2)} \le r_0 \le r_0^{(3)}$, which we call
the {\em principal caustic cone}, that separates regions
with different number of {\em minima}. This is illustrated by Fig.~\ref{fig-zones2}
which is an updated version of Fig.~\ref{fig-caustic-region}, right,
with the numbers of minima and saddle points
shown separately.
The most straightforward proof is based on the stability analysis of the extrema
in three situations with $A_\mu$ lying on each of the three axes
and Proposition~\ref{prop-zones}.
The other cone just separates regions with different numbers of saddle points
and does not affect directly the properties of the global minimum.

\subsection{Geometric reformulation of the search for the global minimum}\label{section-geometric-reform}

Consider again the potential term in (\ref{freeenergy2}):
\be
V = - A_\mu r^\mu + {1\over 2} B_{\mu\nu} r^\mu r^\nu\,.\label{freeenergy3}
\ee
Let us exploit the freedom in definition of $B_{\mu\nu}$ to make
$B_0>0$ and all $B_i<0$. Then $(B^{-1})_{\mu\nu}$ exists, and (\ref{freeenergy3}) can be rewritten as
\be
V = {1\over 2}B_{\mu\nu} (r^\mu-a^\mu) (r^\nu-a^\nu) + V_0\,,\quad
a_\mu = (B^{-1})_{\mu\nu}A^\nu\,,\quad
V_0 = -{1\over 2}(B^{-1})_{\mu\nu}A^\mu A^\nu\,.
\label{freeenergy4}
\ee
Let us now define an {\em equipotential surface} ${\cal M}^C$ as a set of all vectors $p^\mu$ in the Minkowski space ${\cal M}$
such that
\be
B_{\mu\nu} p^\mu p^\nu = B_0 p_0^2 + \sum_i |B_i| p_i^2 = C\,.\label{equipotential}
\ee
One sees that equipotential surfaces exist for $C \ge 0$ and are 3-ellipsoids nested into each other,
with their eigenaxes aligned in the $B^{\mu\nu}$-diagonal frame along the eigenaxes of $B_{\mu\nu}$.

Returning to the potential (\ref{freeenergy4}) we see that $C$ is related
to the values of the potential: $C = 2(V-V_0)$.
Therefore finding points in the orbit space
with the same value of $V$ amounts to finding intersections of the corresponding
${\cal M}^C$ with the forward lightcone $LC^+$.
In particular, to find a local minimum of the potential in the orbit space,
one has to find an equipotential surface that touches $LC^+$
(we say that two surfaces ``touch'' if they have parallel normals at the intersection points).
The {\em global} minimum corresponds to the unique equipotential surface
${\cal M}^{C_\mathrm{min}}$ that only touches but never intersects $LC^+$.

\begin{figure}[!htb]
   \centering
\includegraphics[width=6cm]{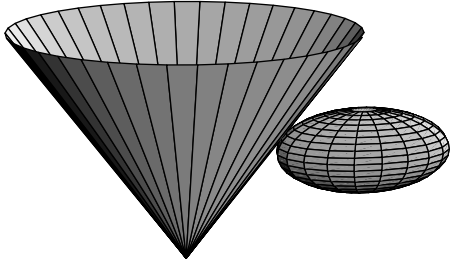}
\caption{A $1+2$-dimensional illustration of the contact between ${\cal M}^{C_\mathrm{min}}$
and $LC^+$. Shown is the case of $a^\mu$ lying outside $LC^+$.}
   \label{fig-contact}
\end{figure}

Thus, the geometric strategy for the minimization of the potential is the following:
\begin{itemize}
\item
Construct a family of 3-ellipsoids ${\cal M}^C$ at the base point $a^\mu$;
\item
Find the unique 3-ellipsoid ${\cal M}^{C_\mathrm{min}}$ that merely touches $LC^+$ but never intersects
it;
\item
The contact point or points give the values of $r^\mu$ (hence, $\psi_1$ and $\psi_2$) that minimize
the potential.
\end{itemize}
To facilitate the visualization, Fig.~\ref{fig-contact} shows a $1+2$-dimensional analogue
of the contact between ${\cal M}^{C_\mathrm{min}}$ and $LC^+$.
In this particular example, $a^\mu$, which is located at the center of the ellipsoid, lies outside $LC^+$.

Alternatively, using (\ref{equipotential}) one can interpret the potential as
the distance squared from the point $a^\mu$ in the Euclidean metric diag$(B_0,\,|B_1|,\,|B_2|,\,|B_3|)$.
The minimization problem is then reformulated as
a search for points on $LC^+$ that are closest to $a^\mu$ in this metric.
Since the forward lightcone $LC^+$ (together with its interior) is a concave region,
this representation immediate leads to the following conclusion:
the necessary condition for the existence of a degenerate minimum is
that $a^\mu$ lies inside $LC^+$: $a^\mu a_\mu > 0$. Later, in Section~\ref{section-spont-viol} we will
give necessary and sufficient conditions for this to happen.

The two geometric constructions described above, the ones based on the equipotential
surfaces and on the caustic cones,
are related to each other in the same manner as a planar curve to its evolute.
To illustrate this relation, let us consider a simple planar problem:
find on the unit circle points of local minima of the ``potential''
$V = B_{ij} (n_i-a_i) (n_j-a_j)$ with $B_{ij} = \mbox{diag}(B_1,\,B_2)$, $B_1\not = B_2$, $B_i>0$.
In the coordinates $\tilde{n}_i = (\sqrt{B})_{ij} n_j$, the ``potential''
becomes $V = |\tilde{n}_i-\tilde{a}_i|^2$, while the unit circle is transformed into an ellipse.
One can easily verify that the number of local minima depends on whether point $\tilde{a}^i$
lies inside the evolute of this ellipse.

\section{Potentials stable in a weak sense}\label{section-weak}

Let us now discuss how the above constructions change for a potential stable in a weak sense,
i.e. a potential whose $V_4$ can have flat directions in the $\psi_i$-space,
along which the potential is stabilized by the $V_2$ term.

A flat direction of the quartic part of the potential in the $\psi_i$-space corresponds
to a vector $r^\mu \in LC^+$ in the orbit space
such that $B_{\mu\nu}r^\mu r^\nu = 0$.
Such an $r^\mu$ must be an eigenvector of $B^{\mu\nu}$, see Appendix~\ref{section-app-positive}.
Let us first assume that there is only one such direction.
Aligning it with the first axis, one can diagonalize $B^{\mu\nu}$ in the ``transverse space'',
bringing it to the following generic form:
\be
B^{\mu\nu} =
\left(\begin{array}{cccc}
B_0 + \delta B & \delta B & 0 & 0\\
\delta B & -B_0 + \delta B & 0 & 0\\
0 & 0 & -B_2 & 0\\
0 & 0 & 0 & -B_3 \end{array}\right)\,,\quad \mbox{with} \quad
\delta B \ge 0\,,\ B_0 > B_2,\,B_3\,.\label{Bi2}
\ee
Note that in contrast to the potentials stable in a strong sense, (\ref{Bi}), this $B^{\mu\nu}$
cannot be diagonalized by an $SO(1,3)$ transformation.
Indeed, a boost along the first axis with ``rapidity'' $\eta$ leads to the same
$B^{\mu\nu}$ as (\ref{Bi2}) but with redefined $\delta B \to e^{-2\eta} \delta B$
(see Appendix~\ref{section-app-B}).
If $\delta B \not = 0$, then $B^{\mu\nu}$ is not diagonalizable.

To find the number of extrema in this case, one can again start with system (\ref{lagrange})
but instead of considering fixed $A_0$ sections in the $A_\mu$ space one
can fix one of its lightcone components.
Let us introduce the lightcone decomposition of any four-vector:
\be
p^\mu = p_+ n_+^\mu + p_- n_-^\mu + p_\perp^\mu\,,\quad n_\pm^\mu = (1,\,\pm 1,\, 0,\, 0)\,,
\ee
where $p_\perp^\mu = (0,\,0,\,p_2,\,p_3)$.
The lightcone coordinates $p_\pm$ are related to the zeroth and first coordinates $p_0$, $p_1$
as $p_\pm = (p_0\pm p_1)/2$. Then, system (\ref{lagrange}) can be rewritten as
\bea
\left\{
\begin{array}{l}
(B_0 - \lambda)r_- = A_-\,,\\
(B_0 - \lambda)r_+ + 2\delta B\,r_- = A_+\,,\\
(B_i - \lambda)r_i = A_i\,,\quad i = 2,3\,,\\
4 r_+ r_- = r_2^2 + r_3^2\,.
\end{array}\right.\label{lagrange3}
\eea
The condition that the quadratic part of the potential $V_2 = - A_\mu r^\mu$ increases
along the flat direction of $V_4$ given by $n_+^\mu$ implies that $A_- <0$.
Repeating the geometric analysis described in full detail for the potential stable in a strong sense,
we arrive at the following conclusion:
\begin{itemize}
\item
If $A^\mu$ lies inside the past lightcone $LC^-$, then system (\ref{lagrange3}) has no solution.
The global minimum is at the origin.
\item
If $A^\mu$ lies outside $LC^-$ (but still with $A_- <0$), then system (\ref{lagrange3}) has
a unique solution.
The corresponding unique extremum of the potential is the global minimum.
\end{itemize}

Suppose now that there are more than one flat direction of $V_4$.
Let us pick up two such distinct vectors $r_1^\mu,\, r_2^\mu \in LC^+$,
both eigenvectors of $B_{\mu\nu}$:
$$
B_{\mu\nu} r_1^\nu = \lambda_1 r_{1\mu}\,,\quad
B_{\mu\nu} r_2^\nu = \lambda_2 r_{2\mu}\,.
$$
Since $r_1^\mu r_{2\mu} \not = 0$, one obtains $\lambda_1 = \lambda_2$.
Then, using the freedom in definition of $B_{\mu\nu}$, one can always set this
common eigenvalue to zero. Then for any linear combination of these two vectors one gets
\be
B_{\mu\nu} (\alpha r_1 + \beta r_2)^\nu = 0\,. \label{lincomb}
\ee
Consider now such an $SO(1,3)$ transformation that makes $r_1^\mu \propto n_+^\mu$
and $r_2^\mu \propto n_-^\mu$. Then, $B_{\mu\nu}$ that satisfies (\ref{lincomb})
takes the following generic form:
\be
B^{\mu\nu} =
\left(\begin{array}{cccc}
0 & 0 & 0 & 0\\
0 & 0 & 0 & 0\\
0 & 0 & -B_2 & 0\\
0 & 0 & 0 & -B_3 \end{array}\right)\,,\quad \mbox{with} \quad
0 \ge B_2,\,B_3\,.\label{Bi3}
\ee
This form is diagonal (note also that it is equivalent to (\ref{Bi2}) with $\delta B =0$),
so one can again switch to the fixed $A_0$ sections.
Since $-A_\mu r^\mu$ must increase along $n^\mu_\pm$,
$A_0 < 0$. Repeating the same analysis as in Section~\ref{section-extrema}, one finds that
the potential has a unique non-trivial extremum
if $|A_1|<|A_0|$ and $A^\mu A_\mu <0$.

So, a potential stable in a weak sense is similar to the
potentials stable in a strong sense with $A_0<0$ in the frame with diagonal $B_{\mu\nu}$.
It can have no more than one non-trivial extremum,
which is then necessarily the global minimum.

\section{Symmetries and their violation}\label{section-symmetries}

\subsection{Explicit symmetries}

As explained in Section~\ref{section-reparametrization}, the free energy
remains invariant under an appropriate simultaneous transformation of the order parameters
$\psi_i$ {\em and} the coefficients. It can happen, however, that
the free energy is invariant under some specific transformation
of $\psi_i$ (or the coefficients) {\em alone}. We call this symmetry
an {\em explicit symmetry} of the free energy.

In the orbit space, this symmetry corresponds to such a map of the Minkowski space
${\cal M}$ that leaves invariant, separately,
$B_{\mu\nu} r^\mu r^\nu$, $A_\mu r^\mu$ and $K_\mu \rho^\mu$.
The notion of explicit symmetry is invariant
under the Lorentz group of the orbit space transformations;
so, any $SO(1,3)$ transformation leaves a given model in the same
symmetry class.

In simple cases the presence of a symmetry can be evident
from a direct inspection of the free energy functional, see e.g. examples
in Section~\ref{section-examples}.
In fact, in many concrete applications, the Landau potential is even constructed
in such a way that some symmetry is preserved.
In the general case, however, a non-evident hidden symmetry can exist
even in complicated forms of the free energy, without being
easily noticeable.
So, one needs a reparametrization-invariant criterion
that can help recognize the presence of a symmetry
using only $K_\mu$, $A_\mu$, and $B_{\mu\nu}$.
In addition, it would also useful to know
what this symmetry is.

Both questions are answered by the following Proposition:
\begin{proposition}\label{prop-classification-symmetries}
Suppose that the free energy (\ref{freeenergy}) is explicitly invariant under some
transformations of $r^\mu$.
Let $G$ be the maximal group of such transformations.
Then:\\
(a) $G$ is non-trivial if and only if there exists an eigenvector
of $B_{\mu\nu}$ orthogonal both to $A_\mu$ and $K_\mu$;\\
(b) group $G$ is one of the following groups: $Z_2$, $(Z_2)^2$,
$(Z_2)^3$, $O(2)$, $O(2)\times Z_2$, or $O(3)$,
and depends on the number of the eigenvectors of $B_{\mu\nu}$ to which $A_\mu$ and $K_\mu$ are orthogonal,
and on whether $B_{\mu\nu}$ has degenerate eigenvalues.
\end{proposition}
\begin{proof}
Let us start with the potential stable in a strong sense.
Let us denote the group of all explicit symmetries of $B_{\mu\nu}$, $A_\mu$, and $K_\mu$
by $G_B$, $G_A$, and $G_K$, respectively.
Obviously,
\be
G = G_B \cap G_A \cap G_K\,.
\ee
The group of explicit symmetries is necessarily a subgroup of the $O(3)$ transformation
group of the 3-dimensional space in the $B_{\mu\nu}$-diagonal frame;
so one can switch to the spacelike parts only ($B_{ij}$, $A_i$, $K_i$).

Consider now $G_B$.
If all spacelike eigenvalues of $B_{\mu\nu}$ are different, then its only symmetries are
reflections of each of the spacelike eigenaxes, which generate group $G_B=(Z_2)^3$.
If two eigenvalues coincide, then $G_B = O(2)\times Z_2$, and if all three of them are
equal, then $G_B=O(3)$.
Note that in all of these cases the following statement holds: if some $Z_2$ group
is a subgroup of $G_B$, then its generator flips
the direction of some eigenvector of $B_{ij}$.

Similarly, $G_A$ is $O(2)$ (rotations around the axis defined by $A_i$),
if $A_i$ is a non-zero vector, and $O(3)$ otherwise.
The same holds for $K_i$, the only difference being the direction of the axis.
If we want $G$ to be non-trivial, then the lowest possible symmetry of $A_i$ and $K_i$
together (given by a $Z_2$ group) must be also a symmetry of $B_{ij}$,
i.e. it must flip one of the eigenvectors of $B_{ij}$.
In other words,
both $A_i$ and $K_i$ are orthogonal to this eigenvector.
Since this purely spacelike eigenvector is also the eigenvector of $B_{\mu\nu}$,
we arrive at the first statement of this Proposition.

Detailed classification depends on the {\em number} of eigenvectors of $B_{ij}$
that are orthogonal to $A_i$ and $K_i$.
\begin{itemize}
\item
If $A_i$ and $K_i$ are orthogonal to all three eigenvectors ($A_i=K_i=0$), then $G=G_B$.
\item
If $A_i$ and $K_i$ are orthogonal to two eigenvectors ($A_i \| K_i$ and are
themselves eigenvectors of $B_{ij}$), then $G = (Z_2)^2$ or $O(2)$.
\item
Finally, if there is only one eigenvector of $B_{ij}$ orthogonal both to $A_i$ and $K_i$,
then the symmetry group is $Z_2$.
\end{itemize}
For a potential stable in a weak sense, we first note that the eigenvectors
of $B^{\mu\nu}$ are either the lightcone vectors or purely spacelike
eigenvectors.
Since $K^\mu$ lies inside $LC^+$, in cannot be orthogonal
to any lightcone vector. Therefore, one has to check the above conditions
only for the ``transverse'' eigenvectors,
which reduces the above list of possible symmetry groups to $Z_2$, $(Z_2)^2$, $O(2)$.
\end{proof}

The necessary and sufficient condition formulated in the first part of this Proposition
can be written in a reparametrization-invariant way. The method is essentially the same
as in \cite{group,nishi} and is based on a simple observation:
if a 3-vector $a_i$ is orthogonal to some eigenvector of a real symmetric
matrix $b_{ij}$, then the triple scalar product of vectors $a_i$,
$b_{ij} a_{j}$, and $b_{ij} b_{jk} a_{k}$ is zero.
In Minkowski space we introduce
\be
K_{0\mu} \equiv K_\mu\,,\quad
K_{1\mu} \equiv B_{\mu}{}^\nu K_\nu\,,\quad
K_{2\mu} \equiv (B^2)_{\mu}{}^\nu K_\nu\,,\quad
K_{3\mu} \equiv (B^3)_{\mu}{}^\nu K_\nu\,,
\ee
where $B^k$ is the $k$-th power of $B_{\mu\nu}$.
The same series can be written for $A_\mu$.
For any four four-vectors $a^\mu$, $b^\mu$, $c^\mu$, and $d^\mu$ we introduce
the short-hand notation
$$
(a,b,c,d) \equiv \epsilon_{\mu\nu\rho\sigma} a^{\mu} b^{\nu} c^{\rho} d^{\sigma}\,.
$$
Then the condition ``there exists an eigenvector of $B_{\mu\nu}$ orthogonal to $K_\mu$''
can be written as
\be
(K_{0}, K_{1}, K_{2}, K_{3}) = 0\,.\label{epsilon4k}
\ee
Note that since $K^\mu$ always lies inside the future lightcone, it can be orthogonal
only to spacelike eigenvectors of $B_{\mu\nu}$, which is exactly what is needed.
Then, the statement of Proposition~\ref{prop-classification-symmetries}a can be reproduced if
we accompany (\ref{epsilon4k}) with a similar condition for $A_\mu$,
\be
(A_{0}, A_{1}, A_{2}, A_{3}) = 0\,,\label{epsilon4m}
\ee
and with the condition that these two 4-vectors be orthogonal to the {\em same} eigenvector of $B_{\mu\nu}$,
for example:
\be
(A_{0}, A_{1}, A_{2}, K_{0}) = 0\,,\label{epsilon3m1k}
\ee

Conditions (\ref{epsilon4k})--(\ref{epsilon3m1k}) can be straightforwardly checked in any basis
once $B_{\mu\nu}$, $A_\mu$, and $K_\mu$ are known.
Thus, the presence of any hidden symmetry can be verified without
the need to find this symmetry explicitly.

\subsection{Symmetries of the potential vs. symmetries of the free energy}\label{section-symm-pot-lang}

Explicit symmetries of the entire free energy
depend on $B_{ij}$, $A_i$, and $K_i$, while the symmetries of the potential
depend only on $B_{ij}$ and $A_i$. Therefore, it might happen that
the potential has a larger symmetry group than the entire free energy.
A simple example is
\be
F = \kappa \left(|\vec{D} \psi_1|^2 + |\vec{D} \psi_2|^2 \right)+
16\lambda \left(|\psi_1|^2-v^2\right)^2 + \lambda \left(|\psi_2|^2-4v^2\right)^2\,.
\ee
The potential here is symmetric under $\psi_2 \leftrightarrow 2\psi_1$, while
the gradient term is not.

The two notions, i.e. the symmetry of the potential or of the entire free energy,
play different roles.
When one seeks for the minimum of the Landau potential, the coefficients in the
gradient term ($K^\mu$) are irrelevant.
However, the symmetry of the spectrum of small
oscillations of the order parameters above the ground state
is the one of the entire free energy functional.

\subsection{Spontaneous breaking of an explicit symmetry}\label{section-spont-viol}

Even if the Landau potential is invariant under some transformation of $\Phi$,
the values $\langle \Phi\rangle$ that minimize it do not necessarily
have to preserve the same symmetry. In the orbit space, if the Landau potential
is invariant under a group $G$ of transformation of $r^\mu$,
then the position of the global minimum might be invariant only under
a proper subgroup of $G$.
In such situations one speaks of spontaneous breaking of the symmetry.
Since the {\em set} of all global minima is invariant
under the full explicit symmetry group $G$, the spontaneous breaking
of an explicit symmetry always leads to degenerate global minima.

For our problem, several results follow immediately from Proposition~\ref{prop-two-minima}:
\begin{enumerate}
\item
The global minimum can be only twice degenerate.
\item
Minima that preserve and violate any discrete symmetry cannot coexist.
\item
The maximal breaking of any discrete symmetry consists in removing
only one $Z_2$ factor: $(Z_2)^{k} \to (Z_2)^{k-1}$, with $k=1,2,3$.
\end{enumerate}

In addition, in \cite{minknew} it was proved that
the twice degenerate global minimum of Landau potential
with quadratic and quartic terms is {\em always} realized
via spontaneous breaking of some explicit $Z_2$ symmetry of the potential
(but not necessarily of the entire free energy!).

Let us now consider the question {\em when} a given explicit symmetry is broken,
focusing on the discrete symmetry case.

First of all, the global minimum must be degenerate. This immediately
leads to the conclusion that the spontaneous violation can take place only in potentials
stable in a strong sense, and in addition, only when $A_\mu$ lies inside
the principal caustic cone.
To make the discussion concrete, consider $A_\mu$ and $K_\mu$ in the $B_{\mu\nu}$-diagonal frame.
Suppose that all $B_i$ are distinct and the components $A_3 = K_3 = 0$,
while the other components are non-zero.
Then, the free energy has an explicit $Z_2$ symmetry generated by reflections of the third coordinate.
This explicit symmetry is conserved, if the global minimum is at
$r^\mu = (r_0,\,r_1,\,r_2,\,0)$,
and it is spontaneously broken if the two degenerate global minima
are at $r_\pm^\mu = (r_0,\,r_1,\,r_2,\,\pm r_3)$ with $r_3 \not = 0$.

Let us now recall the ``shrinking ellipsoid'' construction of Section~\ref{section-extrema}.
A degenerate extremum implies that two distinct points $n_{i\,\pm}$, when inserted in
system (\ref{lagrange2}), give the same point $A_i = (A_1,\,A_2\,,0)$ at the same $r_0$.
This happens only when $r_0= r_0^{(3)}$ and the planar vector $(A_1,\,A_2)$ lies inside
the ellipse with semiaxes
$$
A_0{|B_1-B_3| \over B_0-B_3}\,,\quad A_0{|B_2-B_3| \over B_0-B_3}\,.
$$
Besides, as we prove in Appendix~\ref{section-app-condition},
in order for this extremum to be minimum, $B_3$ must be the largest
(i.e. the closest to $B_0$) eigenvalue among all $B_i$.
Thus, one arrives at the following necessary and sufficient
reparametrization-invariant criterion
for the spontaneous violation of a $Z_2$ symmetry (along the third axis):
\be
B_3>B_1,\,B_2\quad \mbox{and}\quad
{A_1^2 \over (B_3-B_1)^2} + {A_2^2 \over (B_3-B_2)^2} < {A_0^2 \over (B_0-B_3)^2}\,.
\label{violation}
\ee
It immediately follows from here that $a^\mu$ defined in Section~\ref{section-geometric-reform}
lies inside the forward lightcone.

Finally, we would like to stress one important point.
What is relevant for the whole discussion is the group $G$ of explicit symmetries and its
reduction upon symmetry breaking, but not the particular {\em realization} of the transformations of this group.
For example, the free energy can be symmetric under $\psi_i \to \psi_i^*$ transformation or
under $\psi_1\leftrightarrow \psi_2^*$ transformation. These seemingly distinct
$Z_2$ symmetries are, in fact, just different realizations of the {\em same symmetry class}.
This becomes evident in the orbit space, as the former transformation corresponds
to the flip of the second axis, while the latter one corresponds
to the flip of the third axis; so, both models are related by
a reparametrization transformation.
Therefore, all properties of spontaneous violation of these two particular
sorts of the $Z_2$ symmetry are in fact identical.

Thus, there exists a reparametrization-invariant class of $Z_2$-symmetric models,
a reparametrization-invariant class of $(Z_2)^2$-symmetric models, etc.
Our analysis applies to {\em all} particular realizations of a given symmetry in contrast.
This rather simple fact illustrates the usefulness of considering the most general GL model
and might lead to establishing of direct links between seemingly unrelated models.

\section{Phase diagram and phase transitions}\label{section-phase-diagram}

\subsection{Phase diagram}

The results obtained in the previous sections allow us to explicitly describe the phase
diagram of the most general model with two order parameters in the mean-field approximation.
We do this by classifying the phases according to the symmetries
of the free energy functional and to the properties of the ground state.
For definiteness, we again sort the eigenvalues $B_i$ as
$B_1 \le B_2 \le B_3$.
\begin{itemize}
\item
Potential is stabilized by the quadratic term ($A^\mu$ lies inside $LC^-$).
The global minimum is at the origin, $\psi_1=\psi_2=0$;
this is the high-symmetry phase.
\item
Potential stable in a weak sense. The global minimum is always
non-degenerate and preserves any explicit symmetry of the free energy.
\item
Potential stable in a strong sense. $B^{\mu\nu}$ can be diagonalized, and one can
work with its spacelike part.
\begin{itemize}
\item
All $B_i$ are distinct.
\begin{itemize}
\item
$A_i$ and $K_i$ are generic vectors (not orthogonal to any eigenvector of $B_{ij}$).
No explicit symmetry is present.
There can be one or two non-degenerate minima, depending whether $A_i$ lies inside
the principal caustic region.
\item
$A_i$ and $K_i$ are both orthogonal to one eigenvector of $B_{ij}$.
The explicit symmetry group is $Z_2$. The ground state can either break or preserve this symmetry.
The symmetry is broken if it is the third axis that $A_i$ and $K_i$ are orthogonal to (i.e. $A_3=K_3=0$) and if
condition (\ref{violation}) is satisfied.
\item
$A_i$ and $K_i$ are both parallel to the same eigenvector of $B_{ij}$.
The explicit symmetry group is $(Z_2)^2$.
The ground state can either preserve this symmetry or break it to $Z_2$.
The criterion of the symmetry breaking is the same, (\ref{violation}),
apart from the fact that now one of $A_1$, $A_2$ is zero.
\item
$A_i$ and $K_i$ are orthogonal to all three eigenvectors of $B_{ij}$
(i.e. $A_i=K_i=0$, $A_0>0$). The explicit symmetry group is $(Z_2)^3$.
The global minimum is always twice degenerate and breaks this symmetry
to $(Z_2)^2$.
\end{itemize}
\item
Two eigenvalues among $B_i$ coincide. Case $B_1 < B_2=B_3$. \\
The principal caustic cone reduces to a segment along the first axis.
\begin{itemize}
\item
$A_i$ is not aligned along the first axis.
Then, the global minimum is non-degenerate. If $K_i$ lies in the
$(e_{(1)i},A_i)$ plane, then there is an explicit $Z_2$ symmetry
which is preserved at the minimum.
\item
$A_i$ is the eigenvector of $B_{ij}$ along the first axis ($A_2=A_3=0$).
Then, if condition (\ref{violation}) is satisfied,
there is a continuum (namely, a circle) of degenerate minima,
otherwise, the global minimum is non-degenerate.
If, in addition $K_2=K_3=0$, then the explicit symmetry group is $O(2)$,
and in the case of symmetry breaking, it is broken to $Z_2$.
Instead, if $K_i$ is a generic vector, then the explicit symmetry group is $Z_2$,
and it may be preserved or broken depending on
which minimum in the continuum is selected.
\item
$A_i=K_i=0$, $A_0>0$. The explicit symmetry group is $O(2)\times Z_2$.
There is always a continuum (a circle) of global minima, and the symmetry
is broken to $(Z_2)^2$.
\end{itemize}
\item
Two eigenvalues among $B_i$ coincide. Case $B_1=B_2 < B_3$. \\
The analysis is similar to the case with all distinct $B_i$,
with the following differences:
\begin{itemize}
\item
If $A_i$ and $K_i$ are both parallel to the eigenvector $e_{(3)i}$,
the explicit symmetry group is $O(2)$ and it is always preserved at the global minimum.
\item
If $A_i=K_i=0$, then the explicit symmetry group is $O(2)\times Z_2$, which
is broken at the global minimum to $O(2)$.
\end{itemize}
\item
All three eigenvalues $B_i$ coincide.
\begin{itemize}
\item
For a generic pair of $A_i$ and $K_i$, there is always an explicit $Z_2$ symmetry.
It is promoted to the $O(2)$ symmetry if $A_i$ and $K_i$ are parallel, and
to the $O(3)$ symmetry if $A_i=K_i=0$. The symmetry is always preserved,
apart form the case $A_0=0$, when it is broken to $O(2)$ or $Z_2$.

\end{itemize}
\end{itemize}
\end{itemize}

\subsection{First and second order phase transitions}

A remarkable property of the two-order-parameter model is that
it can have a first-order phase transition
even at zero temperature and in the mean-field approximation.
It is entirely due to coexistence
of two local minima (in the orbit space) with different depths.
If upon continuous change of the coefficients
the relative depth of the two distinct minima changes sign
(the shallower minimum becoming the deeper one),
the system occupying initially the global minimum
becomes metastable and can jump into the new global minimum
via fluctuations or quantum tunneling.

If $B_3$ is non-degenerate, then the surface of first-order phase transitions
in the $A_i$ space is the interior of the ellipse at $r_0^{(3)}$,
i.e. it is given by $A_i = (A_1,\,A_2,\,0)$, where $A_1$, $A_2$
satisfy (\ref{violation}).
The border of this ellipse,
$$
A_3 = 0\,,\quad
{A_1^2 \over (B_3-B_1)^2} + {A_2^2 \over (B_3-B_2)^2} = {A_0^2 \over (B_0-B_3)^2}\,,
$$
is the critical line, at which second-order phase transition takes place.
If $B_3$ is degenerate, $B_1<B_2=B_3$, then
the second order phase transition takes place at isolated points
$$
A_3 = A_2 = 0\,,\quad
A_1 = \pm A_0 {B_3-B_1 \over B_0-B_3}\,,
$$
and the points of the first-order transitions form a linear segment
between them.
Finally, if all $B_i$ are degenerate, then there is a single critical point
at the origin, $A_i = 0$, and there is no first order phase transition.

Reconstruction of critical surfaces/lines in the $A_\mu$-space is obvious.

\subsection{Critical properties: an example}

It appears plausible that all the mean-field critical exponents
of the general 2OP GL model are of geometric nature
and can be calculated without
the knowledge of the exact position of the global minimum.
Here, we do not explore this issue in full detail, but just illustrate
it with one example: calculation of the critical exponent of the
correlation length with the aid of differential geometry.

\begin{figure}[!htb]
   \centering
\includegraphics[width=12cm]{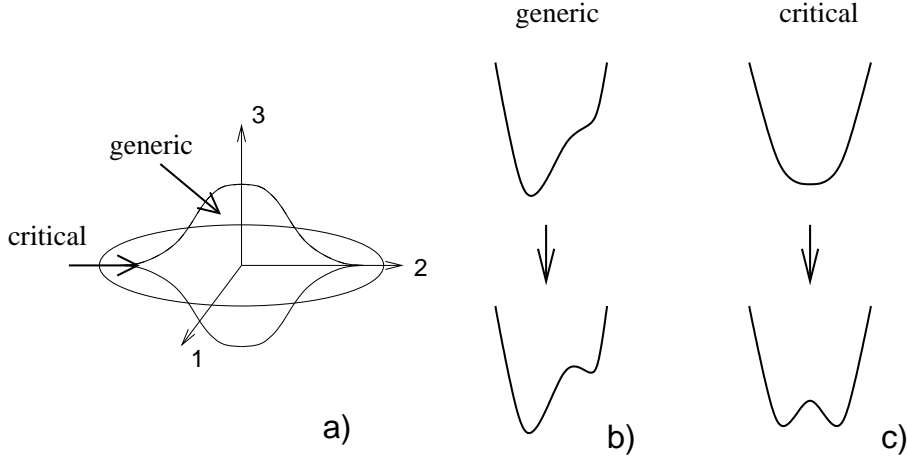}
\caption{(a) Schematic view of the principal caustic region.
The two arrows enter it via a generic or a critical point.
(b) Schematic change of the potential upon a generic entrance into the principal caustic cone.
(c) Schematic change upon the entrance via a critical point.}
   \label{fig-bifurcation}
\end{figure}

Let us fix $B^{\mu\nu}$ and move $A^\mu$ continuously in its parameter space.
When it crosses the principal caustic cone (see Section~\ref{section-principal}),
the number of minima changes. Two scenarios are possible, Fig.~\ref{fig-bifurcation}a.
If $A^\mu$ enters the principal caustic cone through a generic point,
then an additional local minimum appears together with an additional saddle point,
as it is shown schematically in Fig.~\ref{fig-bifurcation}b.
This bifurcation does not involve the global minimum.
However, if $A^\mu$ enters the principal caustic cone through
any of the critical points, then it is the global minimum
that bifurcates into minimum-saddle-minimum sequence, see Fig.~\ref{fig-bifurcation}c.

When $A^\mu$ approaches the critical surface/line, the corresponding eigenfrequency
decreases and turns zero exactly at the critical surface.
If the distance from $A_i$ to the critical surface is
$\varepsilon \to 0$, the eigenfrequency associated with this bifurcation
decreases as
\be
\omega^2 \propto \varepsilon^\delta\,.
\ee
Here we used the fact
that the Jacobian of the map of non-Goldstone modes in $\psi_i$-space
to the surface of $LC^+$ is regular, if the extremum is not at the origin,
see Eq.~(\ref{jacobian}).
The correlation length then behaves as $r_c \propto \varepsilon^{-\delta/2}$.
We argue that the value of $\delta$ is of geometric nature and can be calculated
without the knowledge of the exact position of the minimum.

Let us first note that in the case of a single order parameter $\delta = 1$
simply because the eigenfrequency $\omega^2$ is linearly proportional to the coefficient $a$
in the potential (\ref{landau}). In the two-order-parameter case,
due to the higher dimensionality of the $A^\mu$-space, one can
approach a critical point from different directions.

\begin{figure}[!htb]
   \centering
\includegraphics[width=8cm]{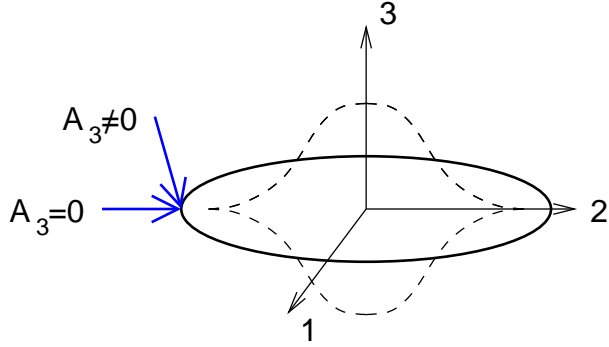}
\caption{(Color online) The principal caustic region in the $A_i$ space (dashed line) and its border, the critical line
(think solid line). The critical exponent $\delta$ depends on whether $A_i$ approaches the critical
line always staying in the plane $A_3=0$ or from outside the plane.}
   \label{fig-critical}
\end{figure}

As we described in Section~\ref{section-geometric-reform}, the search for the global minimum
can be reformulated as a search for points lying on $LC^+$ that are closest
to a given point $a^\mu$ in the Euclidean metric diag$(B_0,\,|B_1|,\,|B_2|,\,|B_3|)$.
In Appendix~\ref{section-app-critical}, using a planar example,
we show how to apply differential geometry
to analyze the properties of the potential near a critical point.
We showed, in particular, that the value of $\delta$ depends on the direction
of approach to the critical point:
\bea
\mbox{generic direction}\  & \to & \delta = {2/ 3}\,,\nonumber\\
\mbox{symmetric approach}\  & \to & \delta = 1\,,\label{delta}
\eea
These exponents are robust in the sense that they remain the same for almost all,
in the measure-theoretic meaning,
regular planar curves. It applies also to the second order curves which
share the key property of the generic curves that they have points of no more than
4th-order contact with a circle.

This technique can be extended to higher dimensions leading to the same results.
So, exponents (\ref{delta}) apply to our problem as well, where ``symmetric approach''
is understood as ``$A_i$ lying in the $A_3=0$ plane'', see Fig.~\ref{fig-critical}.

It would be interesting to check the critical properties of all possible phase transition
in 2OP GL model and see how many classes of critical behavior it can incorporate.

\section{Examples}\label{section-examples}

Here we illustrate the general approach with several simple examples.
Some of them are relevant for condensed-matter problems discussed in literature.

\subsection{Real coefficients}

Consider a free energy functional (\ref{freeenergy}) with all real coefficients.
It implies that the free energy remains invariant under simultaneous transformation
$\psi_i\to \psi_i^*$,
which corresponds in the orbit space to the reflection of the second coordinate:
\be
r^\mu = (r_0,\,r_1,\,r_2,\,r_3) \to (r_0,\,r_1,\,-r_2,\,r_3)\,.
\ee
Consequently, the four-vectors and four-tensor of the coefficients are:
\be
B^{\mu\nu} =
\left(\begin{array}{cccc}
\cdot & \cdot & 0 & \cdot\\
\cdot & \cdot & 0 & \cdot\\
0 & 0 & -B_2 & 0\\
\cdot & \cdot & 0 & \cdot \end{array}\right)\,,\quad
A^\mu\,,K^\mu = (\cdot,\ \cdot,\ 0, \ \cdot)\,,
\ee
where dots indicate generic values. Evidently, conditions (\ref{epsilon4k})--(\ref{epsilon3m1k})
are satisfied.

Just to give a particular example,
consider the free energy functional of a two-gap superconductor
in the dirty limit, see Eq.~(49) in \cite{EXtwogapGurevich}.
Its potential can be rewritten in the reparametrization-invariant way (\ref{freeenergy2})
with
\bea
&&B^{\mu\nu} =
{1\over 2}\left(\begin{array}{cccc}
{b_1 + b_2 \over 2} - b_i & -2b_i & 0 & - {b_1 - b_2 \over 2}\\
-2b_i & 0 & 0 & 0\\
0 & 0 & 0 & 0\\
-{b_1 - b_2 \over 2} & 0 & 0 & {b_1 + b_2 \over 2} - b_i \end{array}\right)\,,\label{gurevich}\\
&& A^\mu = {1\over 2}\left(-a_1-a_2,\, a_i,\, 0,\, a_1-a_2\right)\,.\nonumber
\eea
Here we used the notation of \cite{EXtwogapGurevich}:
coefficients $a_1$, $b_1$ and $a_2$, $b_2$ refer to the properties of the first and second
order parameters, respectively, while $a_i$, $b_i$ describe interaction terms.
The gradient terms considered in \cite{EXtwogapGurevich} are anisotropic,
but they also contain real coefficients. As a result,
$A^\mu$ and $K^\mu$ are orthogonal to the second eigenvector
of $B^{\mu\nu}$. Note that in this example, the eigenvalue $B_2=0$.
In order to find the other eigenvalues, one has to solve the cubic characteristic
equation.

If the position of the global minimum has $\lr{r_2} = 0$, then the symmetry is preserved,
and there is no relative phase between the two order parameters.
If $\lr{r_2} \not = 0$, then the symmetry is spontaneously broken,
and $\lr{\psi_1}$ and $\lr{\psi_2}$ in the ground state have a non-zero relative
phase. In order to find whether spontaneous violation takes place,
one has to diagonalize $B^{\mu\nu}$, find its eigenvalues as well as find
$A^\mu$ in this frame, and then check inequality (\ref{violation}).

\begin{figure}[!htb]
   \centering
\includegraphics[width=6cm]{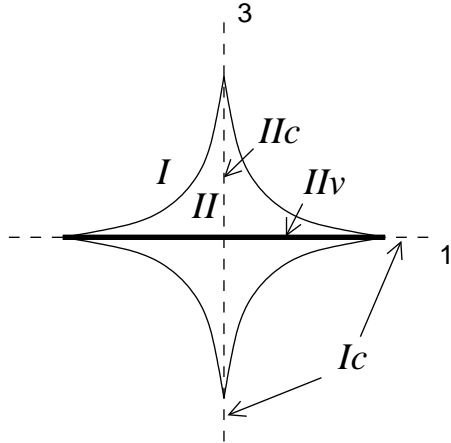}
\caption{Phases of the general Ginzburg-Landau model with two real order parameters on the $A_i$ plane,
classified according to the number of minima and conservation/violation of explicit symmetries.
Shown is the case $B_1<B_3$.}
   \label{fig-real}
\end{figure}

If one requires, in addition, that the order parameters $\psi_i$ themselves be
real, then $r_2 \equiv 0$, and the second axis can be omitted altogether.
The orbit space is then simplified to the forward lightcone in the $1+2$-dimensional
Minkowski space. Repeating the analysis of Section~\ref{section-minima},
one obtains now at most four non-trivial extrema, among which up to two can be minima,
and only one caustic cone.

The phase diagram in this case is simpler.
For example, in Fig.~\ref{fig-real} we show phases in the $A_i$ space for the case $B_3>B_1$
and, for simplicity, we assume that $K_i=0$.
The astroid shown here is the planar analogue of the cusped region
from Fig.~\ref{fig-caustic-region} and Fig.~\ref{fig-zones2}.
It separates the $A_i$ regions corresponding to potentials
with one minimum (phase {\em I}) and two minima (phase {\em II}).
In addition, if $A_i$ lies on the axes, the potental has an explicit symmetry.
Dashed lines corresponds to the cases when the ground state
conserves the symmetry (phases {\em Ic} and {\em IIc}),
while the thick solid line corresponds to the
phase that spontaneously break the discrete symmetry (phase {\em IIv}).

\subsection{No $\psi_1/\psi_2$ mixing in the quartic potential: $b_5=b_6=0$}

The situation simplifies considerably if the quartic potential does not
mix $\psi_1$ and $\psi_2$, i.e. when $b_5=b_6=0$.
In this case $B^{\mu\nu}$ breaks into two $2\times 2$ blocks
and can be easily diagonalized by a boost along the third axis
(see Appendix~\ref{section-app-B} for details)
with ``rapidity''
\be
\eta = {1 \over 4} \log\left({b_1 \over b_2}\right)\,,\label{boost}
\ee
where we assumed $b_1>b_2$,
and by the rotation between the first and second axes by an angle equal to the
half of the phase of $b_4$.
The resulting eigenvalues are:
\be
B_0 = {1 \over 2}(\sqrt{b_1 b_2} + b_3)\,,\quad
B_{1,2} = \pm {|b_4| \over 2}\,,\quad
B_3 = {1 \over 2}(-\sqrt{b_1 b_2} + b_3)\,.
\ee
The condition for the stability in a strong sense of the potential is
\be
b_1>0,\,\quad b_2>0,\, \quad \sqrt{b_1 b_2} + b_3 > |b_4|\,.
\ee
If $A^\mu$ and $K^\mu$ are generic vectors, then the further analysis
proceeds as in the generic case.

\subsection{Interaction only via the $|\psi_1|^2|\psi_2|^2$ term}

Let us assume now that the only interaction between the two OPs
is given by the $b_3|\psi_1|^2|\psi_2|^2$ term.
The four-vector $A^\mu$ can be written as $(A_0,\,0,\,0,\,A_3)$,
so that the potential has an explicit $O(2)$ symmetry.
The eigenvalues of $B^{\mu\nu}$ expressed in terms of the original coefficients are
\be
B_0 = {1 \over 2}(\sqrt{b_1 b_2} + b_3)\,,\quad
B_{1,2} = 0\,,\quad
B_3 = {1 \over 2}(-\sqrt{b_1 b_2} + b_3)\,,
\ee
where for stability we require $\sqrt{b_1 b_2} + b_3 > 0$.
In the frame where $B^{\mu\nu}$ is diagonal, $A^\mu$ takes the form
\be
(\tilde{A}_0,\,0,\,0,\,\tilde{A}_3)=(A_0\cosh\eta + A_3\sinh\eta,\,0,\,0,\,A_3\cosh\eta + A_0\sinh\eta),
\ee
where $\eta$ is given by (\ref{boost}).

If $B_3>0$, i.e. $b_3 > \sqrt{b_1 b_2}$, then the explicit $O(2)$ symmetry is always conserved,
since $B_1$ and $B_2$ are not the largest spacelike eigenvalues.
The global minimum is at $\lr{r^\mu} \propto n_+^\mu$ or $n_-^\mu$,
which corresponds to
$$
\lr{\psi_1} \not = 0\,,\lr{\psi_2} = 0\,,\quad\mbox{or}\quad\lr{\psi_1} = 0\,,\lr{\psi_2} \not = 0\,.
$$
If $B_3<0$ and if $A^\mu$ lies inside the caustic cone,
then this symmetry is spontaneously broken, and there exists a continuum
of global minima with both $\lr{\psi_1} \not = 0$ and $\lr{\psi_2} \not = 0$
and an arbitrary relative phase between them.
The condition that $A^\mu$ lies inside the caustic cone is
$$
\left|{\tilde{A}_3 \over B_1-B_3}\right| < {\tilde{A}_0 \over B_0-B_1}\,,
$$
which, in terms of the original coefficients translates into
\be
a_1{b_3 \over b_1} < a_2 < a_1{b_2 \over b_3}\,.
\ee
Of course, the same bounds can be obtained by direct calculations.

Note that in the case of no interaction at all, $b_3=0$, the
condition for the symmetry violation reads $|A_3|<A_0$. It means
$a_1>0$, $a_2>0$, which is indeed expected.

\section{Solitons}\label{section-solitons}

Two local order parameters can lead to existence of solitons,
i.e. states with non-trivial coordinate dependence
of the mean-field values of the order parameters $\lr{\psi_i}(\vec r)$
stable against small variations $\delta\lr{\psi_i}(\vec r)$
of these OP profiles.
Some particular versions of such solitons have been already described
in literature. For example, in \cite{soliton},
a one-dimensional two-band superconductor
with a simple interband interaction term was considered,
whose ground state corresponded to $\lr{\psi_1}$ and $\lr{\psi_2}$
with zero relative phase.
Then, a typical sine-Gordon soliton was constructed
with the relative phase between the two OPs continuously changing
from zero to $2\pi$ at $x = \pm \infty$, correspondingly.
Similar solitons in the scalar sector of 2HDM were described
in \cite{tomaras}.

Existence of solitons in a given 2OP GL model depends
on the geometry of the potential in the orbit space.
For example, in order to support a one-dimensional soliton similar to
the one described above, the Landau potential must have a certain ``valley''
(i.e. a region of low values of the potential) of non-trivial topology on the forward lightcone $LC^+$,
that would include the global minimum and a saddle point.
At $x \to -\infty$, $\lr{\psi_i}(x)$ approach their values at the global minimum.
As $x$ increases, the corresponding point
$\lr{r^\mu}(x)$ in the orbit space moves away from
the global minimum position, follows some path in the valley and returns again to the global minimum.
Small variations of $\lr{\psi_i}(x)$ would pull this path out of the valley,
increasing its potential energy.

Existence, stability and geometric properties (e.g. dimensionality)
of these solitons are sensitive only to the general structure of the model, and do not
require one to search for the explicit position of the extrema.
Therefore, one can hope to obtain these criteria for a general 2OP GL model
in terms of geometric constructions studied in this paper.

\section{Multicomponent order parameters}\label{section-multi}

So far, we assumed that the order parameters $\psi_1$ and $\psi_2$ are just complex numbers.
However, in many physical situations one introduces {\em multicomponent order parameters}.
Examples include 2HDM, superfluidity in $^3$He, non-conventional superconductivity,
spin-density waves, etc.

The formalism presented above is applicable to these cases as well. In fact, it was first developed
in \cite{mink,minknew} specifically for 2HDM.
Here, we discuss characteristic features that appear in a generic GL model with two $N$-vector
order parameters and a $U(N)$-symmetric potential.

\subsection{Modifications to the formalism}

Let us assume that each $\psi_i$ is an $N$-dimensional complex vector:
$\psi_{i\,\alpha}$, $\alpha = 1,\,\dots,\,N$.
A $U(N)$-symmetric potential must depend on the order parameters only via scalar combinations
$(\psi^\dagger_i \psi_j)$, $i,j,=1,2$, which parametrize the $U(N)$-orbits.
The only difference with the scalar case is that an additional term proportional to
\be
(\psi^\dagger_1 \psi_2) (\psi^\dagger_2 \psi_1) \not = |\psi_1|^2 |\psi_2|^2\label{newterm}
\ee
appears in the potential, with a new independent coefficient $b_3^\prime$ in front.
The definition of $r^\mu$ remains the same; however
\be
r^\mu r_\mu = 4\left[(\psi^\dagger_1 \psi_2) (\psi^\dagger_2 \psi_1) - |\psi_1|^2 |\psi_2|^2\right] \ge 0\,.\label{LCnew}
\ee
Therefore, the orbit space now is not only the surface, but also the {\em interior} of the forward lightcone $LC^+$.
This removes the degree freedom in definition of $B_{\mu\nu}$, making it uniquely defined:
\be
B^{\mu\nu} = {1\over 2}\left(\begin{array}{cccc}
{b_1+b_2 \over 2} + b_3 & -\Re(b_5 + b_6)
    & \Im(b_5 + b_6) & -{b_1-b_2 \over 2} \\[1mm]
-\Re(b_5 + b_6) & b_3^\prime + \Re b_4 & -\Im b_4 & \Re(b_5 - b_6) \\[1mm]
\Im(b_5 + b_6) & -\Im b_4 & b_3^\prime -\Re b_4 & -\Im(b_5 - b_6) \\[1mm]
 -{b_1-b_2 \over 2} & \Re(b_5 - b_6) & -\Im(b_5 - b_6)
 & {b_1+b_2 \over 2} - b_3
\end{array}\right)\,.\label{Bmununew}
\ee
The requirement that the potential is stable in a strong sense implies
that $B_{\mu\nu}$ must be positive definite {\em on and inside} $LC^+$.
This leads not only to $B_0 > B_i$, but also to $B_0>0$.
Note that due to the absence of freedom in $B_{\mu\nu}$,
cases with singular $B_{\mu\nu}$ and with $B_i$ of different signs must be considered as well.

\subsection{Consequences}

Let us discuss the modification of the above analysis due to the multicomponent order parameters.

A new phase appears, which is characterized by a stronger breaking of the symmetry
of the potential. It corresponds to the global minimum of the potential $\lr{r^\mu}$ lying
strictly inside the lightcone $LC^+$. This is possible only when
$\lr{\psi_1}$ and $\lr{\psi_2}$ are not proportional to each other.
In other words, one can always perform a simultaneous ``intra-vector'' $U(N)$ transformation
of both order parameters that makes them
\be
\lr{\psi_1} = \left( \begin{array}{c}
0 \\  \vdots \\ 0 \\ v_1 \end{array}\right) \,,\quad
\lr{\psi_2} = \left( \begin{array}{c}
0 \\  \vdots \\ u \\ v_2 e^{i\xi}
 \end{array}\right) \,,\label{psi-charged}
\ee
where dots indicate zeros. Here, $u$, $v_1$, $v_2$, $\xi$ are real,
and $u$ and $v_1$ must be non-zero
in order for $r^\mu$ constructed from them to lie strictly inside $LC^+$.

Solution (\ref{psi-charged}) with non-zero $u$ preserves only a $U(N-2)$ symmetry,
while a normal solution lying on $LC^+$ and corresponding to $u = 0$
preserves a $U(N-1)$ symmetry.
For example, in the context of the two-Higgs-doublet model ($N=2$)
such a solution corresponds to a complete breaking of the electroweak symmetry group $SU(2)\times U(1)$.
Such phase breaks the electric charge conservation,
and makes the photon massive.

Conditions when this phase appears were established in \cite{mink}.
Since $\lr{r^\mu}$ that corresponds to such a non-symmetric phase is not restricted anymore to lie
on the surface of $LC^+$, the extremum condition of the potential takes a very simple form:
\be
B_{\mu\nu} r^\nu = A_\mu\,. \label{charged}
\ee
If $B_{\mu\nu}$ is non-singular, then solution of (\ref{charged}) always exists and is unique.
If the potential has any additional explicit symmetries, this symmetry is always conserved in
this phase. If $B_{\mu\nu}$ is singular, then depending on $A_\mu$ Eq.~(\ref{charged})
can have an empty set or a continuum of solutions.

Whether the solution of (\ref{charged}) corresponds to a physically realizable extremum
of the potential, depends on whether $a_\mu = (B^{-1})_{\mu\nu}A^\nu$ lies inside $LC^+$.
If it is so, then it can be a minimum or a saddle point.
It is a minimum (and necessarily the global minimum) when $B_{\mu\nu}$ is positive-definite
in the entire Minkowski space, i.e. when all $B_i <0$.

Search for the extrema on the forward lightcone $LC^+$ proceeds in the same
way as before. One again introduces equipotential surfaces ${\cal M}^C$,
but due to fixed eigenvalues $B_i$ their geometry can be different.
A typical ${\cal M}^C$ can now be any 3-quadric:
a 3-hyperboloid, a 3-ellipsoid, a 3-cone, or a 3-paraboloid.
The geometric reformulation of the problem remains unchanged: the search for the global minimum
corresponds to the search for the unique 3-quadric with the base point $a^\mu$
that touches but never intersects the forward lightcone.

As a result, virtually all the statements about the number of extrema and minima, about the
symmetries and their spontaneous violation remain the same.
The only difference is that $r^\mu$ can shift from the surface of $LC^+$ inwards,
and in order for an extremum on $LC^+$ to be a minimum, this shift
must also increase the potential. It means that Lagrange multiplier $\lambda$
in (\ref{lagrange}) must be positive. In fact, the eigenfrequencies $\omega^2$
of oscillations that make $N$-vectors $\psi_1$ and $\psi_2$ non-parallel
are proportional to $\lambda$. In 2HDM, they correspond to the masses of charged
Higgs bosons, \cite{mink,nachtmann}.

\section{Conclusions}\label{section-conclusions}

The aim of this paper was to provide an exhaustive description
of the general two-order-parameter model with
all possible $U(1)$-symmetric quadratic and quartic interaction terms
in the mean-field approximation.
The principal difficulty in the study of this model lies in the fact
that the Landau potential cannot be minimized with straightforward algebra.
Here we showed that despite this difficulty, one can still learn
a lot about the phase structure of this model.
We developed the Minkowski-space formalism based on the reparametrization symmetry
of the model and reformulated the minimization problem in
simple geometric terms.
We then proved several statements concerning the properties
of the model (the number of extrema and minima, symmetries and their violation,
the phase diagram).

The most general 2OP GL model can be viewed as a ``template''
for many particular realizations of the two-order-parameter
model used in various condensed matter problems.
We believe that by considering the most general case one can
gain a more transparent understanding of phenomena taking place
in particular situations, and one might even establish new links
between seemingly unrelated models.

We also note that the general method used in this paper
(consider the model in the most general case,
find the group of reparametrization symmetries, and using it
find the structure behind the model) is very general
and might turn out helpful in other circumstances.

There remain several directions for future work.
First, using dependence of the coefficients on temperature, pressure,
etc., one can trace in detail the sequence of phase transitions
as well as their critical properties in the mean field approximation.
Second, one should study modifications caused by the presence
of external fields (e.g. magnetic fields for two-gap superconductors)
and non-trivial boundary conditions.
Third, one should analyze effects beyond the mean-field approximation,
in particular, study how the symmetries of the model evolve
under the renormalization group flow.
Fourth, one should closely examine the existence, stability and
dynamics of the solitons.
Finally, extension of the approach to models with several order parameters
and/or with matric-valued OPs also appears to be feasible.
\\

I am thankful to Ilya Ginzburg, Otto Nachtmann and to the referees
for discussions and useful comments. This work was supported by
the Belgian Fund F.R.S.-FNRS via the contract of Charg\'e de
recherches and in part by grants RFBR 08-02-00334-a and
NSh-1027.2008.2

\appendix

\section{Manipulation with 4-tensor $B_{\mu\nu}$}\label{section-app-B}

Here we collect some simple facts about the real symmetric 4-tensor $B_{\mu\nu}$.
Let us first give explicit expressions for $B_{\mu\nu}$ with raised or lowered indices:
\be
B^{\mu\nu} =
\left(\begin{array}{cc}
B_{00} & B_{0j} \\
B_{0i} & B_{ij}
\end{array}
\right)\,,\quad
B^{\mu}{}_{\nu} = B_{\mu\alpha} g^{\alpha\nu} =
\left(\begin{array}{cc}
B_{00} & - B_{0j} \\
B_{0i} & - B_{ij}
\end{array}
\right)\,,\quad
B_{\mu\nu} =
\left(\begin{array}{cc}
B_{00} & - B_{0j} \\
- B_{0i} & B_{ij}
\end{array}
\right)\,.\label{lambdamunu3}
\ee
Here $i,j=1,2,3$. Note that $B^{\mu}{}_\nu$ is not symmetric.

Upon an $SO(3)$ rotation, $B_{00}$ remains invariant, while $B_{0i}$ and $B_{ij}$ transform
as real 3-vector and symmetric 3-tensor, respectively.
Upon a boost with ``rapidity'' $\eta$, say, along the first axis, $B^{\mu\nu}$ transforms as:
\be
B^{\mu\nu} =
\left(\begin{array}{cccc}
b_{00} & b_{01} & b_{02} & b_{03}\\
b_{01} & b_{11} & b_{12} & b_{13}\\
b_{02} & b_{12} & b_{22} & b_{23}\\
b_{03} & b_{13} & b_{23} & b_{33}
 \end{array}\right) \quad \to \quad
(B^\prime)^{\mu\nu} =
\left(\begin{array}{cccc}
b'_{00} & b'_{01} & b'_{02} & b'_{03}\\
b'_{01} & b'_{11} & b'_{12} & b'_{13}\\
b'_{02} & b'_{12} & b'_{22} & b'_{23}\\
b'_{03} & b'_{13} & b'_{23} & b'_{33}
\end{array}\right)\,,\nonumber
\ee
where
\bea
b'_{00} &=& {b_{00} - b_{11}\over 2} + {b_{00} + b_{11}\over 2}\cosh 2\eta + b_{01} \sinh 2\eta\,,\nonumber\\
b'_{11} &=& -{b_{00} - b_{11}\over 2} + {b_{00} + b_{11}\over 2}\cosh 2\eta + b_{01} \sinh 2\eta\,,\nonumber\\
b'_{01} &=& {b_{00} + b_{11}\over 2}\sinh 2\eta + b_{01} \cosh 2\eta\,,\nonumber\\
b'_{0a} &=& b_{0a} \cosh \eta + b_{1a} \sinh\eta\,,\quad
b'_{1a} = b_{1a} \cosh \eta + b_{0a} \sinh\eta\,,\quad b'_{ab} = b_{ab}\,,\quad a,b=2,3\,,\nonumber
\eea
The eigenvalues $B_i$ and eigenvectors $e_{(i)}^\mu$ of $B_{\mu\nu}$ are defined according to
\be
B_{\mu\nu} e_{(i)}^\nu = B_i\, g_{\mu\nu} e_{(i)}^\nu\,,\quad\mbox{or equivalently}\quad
B_{\mu}{}^{\nu} e_{(i)\, \nu} = B_i\, e_{(i)\, \mu}\,.\label{eigen}
\ee
The fact that $B_{\mu}{}^{\nu}$ is not symmetric
implies that the some eigenvalues might be, in general, complex.
However, as proved below, positive definiteness of $B_{\mu\nu}$
on the forward lightcone $LC^+$ ensures that they are real.

In the diagonal basis one has:
$$
B_{\mu\nu} =
\left(\begin{array}{cccc}
B_0 & 0 & 0 & 0\\
0 & -B_1 & 0 & 0\\
0 & 0 & -B_2 & 0\\
0 & 0 & 0 & -B_3 \end{array}\right)\,,\quad
B_{\mu}{}^{\nu} =
\left(\begin{array}{cccc}
B_0 & 0 & 0 & 0\\
0 & B_1 & 0 & 0\\
0 & 0 & B_2 & 0\\
0 & 0 & 0 & B_3 \end{array}\right)\,.
$$
If one considers a quadratic form in the space of 4-vectors $p^\mu$ constructed on $B_{\mu\nu}$,
then in the diagonal basis it looks as
$$
B_{\mu\nu} p^\mu p^\nu = B_0 p_0^2 - \sum_i B_i p_i^2\,.
$$
This quadratic form is positive definite in the entire space of non-zero vectors $p^\mu$, if and only if
all $B_i$ are {\em negative}.

\section{Positive definiteness of $V_4$}\label{section-app-positive}

A potential stable in a strong sense was defined as the one whose quartic part $V_4$
is strictly positive definite in the entire space of the order parameters $\psi_i$ except the origin.
In the orbit space it corresponds to $B_{\mu\nu} r^\mu r^\nu$ being positive definite
on the entire forward lightcone $LC^+$ expect the apex. This criterion can be formulated
in terms of the eigenvalues of $B^{\mu\nu}$:
\begin{proposition}\label{appA}
Tensor $B^{\mu\nu}$ is positive definite on the future lightcone expect the apex
if and only if the following conditions
are met:\\
(1) $B^{\mu\nu}$ is diagonalizable by an $SO(1,3)$ transformation,\\
(2) all spacelike eigenvalues $B_i$ are smaller than the timelike eigenvalue $B_0$.
\end{proposition}

\begin{proof}
Obviously, if $B^{\mu\nu}$ satisfies conditions (1) and (2),
then the positive definiteness follows immediately.
So, one needs to prove that these conditions follow from the positive definiteness.

The first step is to prove that the positive definiteness on $LC^+$
implies that all the eigenvalues of $B^{\mu\nu}$ are real.

Suppose, on the contrary, that there exists a pair
of non-zero complex eigenvalues, $b$ and $b^*$,
with respective (and necessarily complex)
eigenvectors $p^\mu$ and $q^\mu$:
$$
{B^{\mu}}_{\nu} p^\nu = b p^\mu\,,\quad
{B^{\mu}}_{\nu} q^\nu = b^* q^\mu\,.
$$
One can show that there can be only one pair of complex eigenvalues, thus, $b$
is non-degenerate.
Since $B^{\mu}{}_{\nu}$ is real, $q^\mu \propto p^{\mu *}$ (and can be taken equal to $p^{\mu *}$).
These eigenvectors are orthogonal, $p^\mu q_\mu = 0$, which follows from the standard
argument due to $b \not = b^*$, and can be normalized so that
$p^\mu p_\mu = q^\mu q_\mu = 1$.

Consider now a non-zero real vector $r^\mu$,
$$
r^\mu = c p^\mu + c^* p^{* \mu}\,,
$$
such that $r^\mu r_\mu = c^2 + c^{* 2} = 2|c|^2\cos(2\phi_c) = 0$.
At fixed $|c|$, four such vector are possible. Take two of them: $r^\mu_{1/4}$ and $r^\mu_{3/4}$,
corresponding to $\phi_c = \pi/4$ and $3\pi/4$.
The quadratic form calculated on these vectors is
$$
B_{\mu\nu} r^\mu r^\nu = b c^2 + b^* c^{* 2}
= 2|b||c|^2\cos\left(2\phi_c + \phi_b\right)
= \mp 2|b||c|^2\sin(\phi_b)
$$
for $r^\mu_{1/4}$ and $r^\mu_{3/4}$, respectively.
Since $b$ is not purely real, $\sin(\phi_b) \not = 0$,
in one of the two cases $B_{\mu\nu} r^\mu r^\nu < 0$,
which contradicts the assumption.

After all the eigenvalues of $B_{\mu\nu}$ are proved to be real, the eigenvectors also
can be chosen all real and orthonormal.
These eigenvectors cannot lie on $LC^+$ (otherwise there would be a flat direction of $V_4$),
so there is one vector inside $LC^+$ with positive norm, norm $p_0^\mu p_{0 \mu} = 1$,
and three spacelike eigenvectors with negative norms $p_i^\mu p_{i \mu} = -1$
for each $i=1,2,3$.
Thus, the transformation matrix $T$ that diagonalizes $B^{\mu\nu}$
is real, and after diagonalization $B^{\mu\nu}$ takes form
$\mathrm{diag}(B_0,\, -B_1,\, -B_2,\, -B_3)$.
Note that transformation $T$ also conserves norm,
so it can be realized as a transformation from the proper Lorentz group.

Now, the requirement that $B^{\mu\nu}$ is positive definite on $LC^+$ reads:
$$
B_0 - (B_1 \sin\theta\cos\phi + B_2 \sin\theta\sin\phi + B_3\cos\theta) > 0
$$
for all $0\le \theta\le \pi$ and $\phi$. This holds when $B_0$ is larger than any $B_i$.
\end{proof}

Let us also see what changes for a potential stable in a weak sense.
First, the statement that the eigenvalues are real and therefore
eigenvectors can be also chosen real remains valid in this case.
However, at least one eigenvector must now lie on the surface of $LC^+$.
This means that $B^{\mu\nu}$ is in general not diagonalizable by the $SO(1,3)$
transformation group. More details are given in Section~\ref{section-weak}.

\section{Necessary condition for the spontaneous violation: explicit calculations}\label{section-app-condition}

Here, we show that the global minimum
of the potential with all distinct $B_i$ and $A_\mu = (A_0,\,A_1,\,A_2,\,0)$
can spontaneously break the $Z_2$ symmetry given by reflections of the third axis,
only if $B_3$ is the largest spacelike eigenvalue:
\be
B_3>B_1,\,B_2\,.\label{appB3}
\ee
We assume, of course, that the vector $A_\mu$ lies inside the caustic cone:
\be
{A_1^2 \over (B_3-B_1)^2} + {A_2^2 \over (B_3-B_2)^2} < {A_0^2 \over (B_0-B_3)^2}\,.\label{appcaustic}
\ee
This will be done by comparing the depth of the potential at the extrema
that conserve and violate this symmetry. We will see that (\ref{appB3})
is necessary for the pair of symmetry violating extrema to be the deepest ones.

If $\lr{r^\mu}$ is an extremum point, then the potential at this point is
$$
V = - A_\mu \lr{r^\mu} + {1 \over 2} B_{\mu\nu} \lr{r^\mu}\lr{r^\nu}
= - {1 \over 2}A_\mu \lr{r^\mu} = - {1 \over 2} B_{\mu\nu} \lr{r^\mu}\lr{r^\nu}\,.
$$
According to (\ref{lagrange2}), the symmetry violating extrema take place at
$$
r_0 = r_0^{(3)} = {A_0 \over B_0-B_3}\,.
$$
The depth of the potential at this point is
$$
|V_3| = {1 \over 2} \left({A_0^2 \over B_0-B_3} + {A_1^2 \over B_3-B_1} + {A_2^2 \over B_3-B_2} \right)\,.
$$
Pick up another extremum (necessarily a symmetry-conserving one). It takes place at another $r_0$,
which we rewrite as $r_0 \equiv r_0^{(3)}\cdot x$.
The depth of the potential at this point is
\bea
|V| & = & {1 \over 2}r_0(A_0 - A_1 n_1 - A_2 n_2) \nonumber\\
&=& {1 \over 2} x \left[{A_0^2 \over B_0-B_3} - {A_1^2 \over (B_0-B_3)-(B_0-B_1)x}
- {A_2^2 \over (B_0-B_3)-(B_0-B_2)x}\right]\,.\nonumber
\eea
Here, $n_1$, $n_2$ are
$$
n_1 = {A_1 \over A_0 -(B_0-B_1)r_0}\,,\quad n_2 = {A_2 \over A_0 -(B_0-B_2)r_0}\,,\quad
n_1^2 + n_2^2 = 1\,.
$$
Note that the last equation, in fact, is the fourth-order equation for $r_0$.

Difference between the two depths can be presented, after some algebra, in the following way:
\be
|V_3| - |V| = {(1-x)A_0 \over 2} \left({A_0 \over B_0-B_3} - {A_1 n_1 \over B_1-B_3}
- {A_2 n_2 \over B_2-B_3}\right)\,.
\ee
The expression in brackets can be rewritten as $\alpha_\mu n^\mu$,
where $n^\mu = (1,\,n_1,\,n_2,\,0)$ and
$$
\alpha^\mu = \left({A_0 \over B_0-B_3},\, {A_1 \over B_1-B_3},\, {A_2 \over B_2-B_3},\, 0\right)
\,.
$$
From the caustic condition (\ref{appcaustic}) one obtains $\alpha^\mu \alpha_\mu > 0$,
i.e. the four-vector $\alpha_\mu$ lies strictly inside $LC^+$.
On the other hand, $n^\mu$ lies on the surface of $LC^+$, and therefore,
$\alpha_\mu n^\mu > 0$. Thus, the sign of the depth difference is given solely by the value of $x$.

If $B_3$ is the largest spacelike eigenvalue, then all symmetry-conserving extrema
correspond to $r_0 < r_0^{(3)}$, i.e. to $x<1$. Therefore, all of them lie above
the symmetry-violating points (and are saddle points, according to Proposition~\ref{prop-two-minima}).
If $B_3$ is not the largest spacelike eigenvalue, then there will necessarily be another
extremum with $r_0> r_0^{(3)}$, which corresponds to $x>1$ and, therefore, lies deeper that
the symmetry-violating points (which are saddle points in this case).

An alternative, somewhat longer way to prove condition (\ref{appB3}) using
geometric properties of the potential was given in \cite{minknew}.

\section{Critical exponent for the distance squared function defined on a planar curve}\label{section-app-critical}

Let $\gamma(t)$ be a regular parametrization of a smooth plane curve,
see e.g. \cite{diffgeom}.
Take a point on this curve, assuming that it corresponds to $t=0$,
and choose the coordinate frame at this point such that axis $x$ is along the tangent
and axis $y$ is along the normal to the curve at this point, see Fig.~\ref{fig-curve}.
The curve then can be parametrized as $\gamma(t) = (X(t),\,Y(t))$, with
\be
X(t) = t\,,\quad Y(t) = {t^2\over 2 R_0} + a_n t^n + o(t^n)\,.\label{curve}
\ee
Here, $R_0$ is the curvature radius of $\gamma$ at $t=0$, while $n>2$ describes
the next higher order term.

\begin{figure}[!htb]
   \centering
\includegraphics[width=5cm]{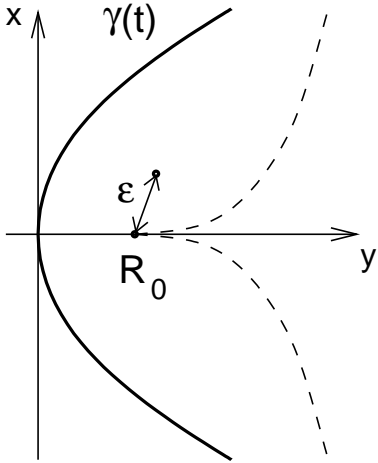}
\caption{Regular curve $\gamma(t)$ and its evolute (dashed line); see text.}
   \label{fig-curve}
\end{figure}

Now, select a point $\vec r = (x,y)$ on this plane and calculate the distance squared
from this point to the points of the curve, $\rho^2(t) \equiv (X(t)-x)^2 + (Y(t)-y)^2$.
This function has points of extrema at some values of $t$.
For a generic point $\vec r$, $\rho^2(t)$ will have a generic value at $t=0$.
However, if $\vec r$ lies on the $y$ axis, then $\rho^2(t)$ has a maximum or minimum at $t=0$.
At a special (``critical'') point along this axis, $\vec r = (0,\,R_0)$,
$$
\rho^2(t) = X^2(t) + (Y(t)-R_0)^2 \approx \mbox{const} - 2R_0 a_n t^n + {t^4 \over 4R_0^2}\,.
$$
Therefore, the osculating circle has at 3-fold or 4-fold contact with the curve $\gamma$,
depending on whether $n=3$ or $n\ge 4$. As the point $\vec r$ moves along the $y$ axis and passes
through $(0,\,R_0)$, a bifurcation takes place of the function $\rho^2(t)$.

There are two possibilities to consider.
If $\gamma(0)$ is a generic point on a generic smooth curve, then expansion (\ref{curve})
starts from $n=3$, and at small $t$, $\rho^2(t)$ has a minimum and a maximum near at $t=0$,
both of which cannot be the global ones.
Instead, if $\gamma(0)$ is an apex of the curve (this is the situation
shown in Fig.~\ref{fig-curve}), then expansion (\ref{curve})
starts from $n=4$, and at this critical point the minimum of $\rho^2(t)$
splits into minimum/maximum/minimum sequence.
This is the only type of bifurcation the global minimum of the function
$\rho^2(t)$ can participate, and we now focus on it.

Let us shift $\vec r$ away from the critical point by a small amount,
$\vec r = (\varepsilon_x,\,R_0 + \varepsilon_y)$, and recalculate $\rho^2(t)$:
\be
\rho^2(t) \approx \mbox{const} + t^4\left({1 \over 4R_0^2} - 2 a_4 R_0\right) - 2\varepsilon_x t
-{\varepsilon_y t^2 \over R_0}\,.\label{rho2}
\ee
For a generic curve $\gamma(t)$ (the second order curves included) the coefficient in front of $t^4$ is non-zero.
Finding the minimum of (\ref{rho2}), expanding $\rho^2$ near it, and extracting the coefficient
in front of the quadratic term $(t-t_{min})^2$, which should behave as $\propto |\varepsilon|^\delta$,
gives us the value of $\delta$. One can easily find that it depends
on the {\em direction of approach} to the critical point:
\bea
\mbox{generic direction}\  (\varepsilon_x \not = 0) & \to & \delta = {2/ 3}\,,\nonumber\\
\mbox{symmetric approach}\  (\varepsilon_x = 0) & \to & \delta = 1\,,\label{delta2}
\eea
The latter case, by construction, effectively corresponds to the standard
one-order-parameter Ginzburg-Landau model.

This study can be generalized to the $n+1$-dimensional case.
Given a smooth $n$-manifold $\gamma(t_1,\dots,t_n)$, one can choose
the coordinate frame a point of this manifold, that would correspond to the global
minimum, align the coordinate axes with the eigenvectors
of the quadratic term, and parametrize the manifold as
\be
X_1 = t_1\,,\dots,\, X_n = t_n\,,\quad Y \approx {t_1^2\over 2R_1} + \cdots + {t_n^2\over 2R_n} + a_4 t_1^4\,.
\ee
Here, we labeled the axes according to $R_1 < R_2 < \cdots < R_n$,
the bifurcation we study is at $\vec r = (0,\dots,0,R_1)$.
The calculations can be repeated giving the same result: if $\varepsilon_x \not = 0$, then $\delta = 2/3$,
otherwise $\delta = 1$. Now $\varepsilon$ must be understood as the distance from the closest among
the cusp points of the evolute.
This result does not depend on the particular
shape of the manifold, since it is essentially driven by the $4$-th order nature
of the bifurcation point.

\end{document}